\documentclass[reqno,11pt,twoside]{article}
\usepackage{amssymb,amsthm,amsmath,amsfonts}
\usepackage{epsfig}

\theoremstyle{plain}
\begingroup
 
\newtheorem{thm}{Theorem}[section] 

\newtheorem{lemma}[thm]{Lemma}
\newtheorem{prop}[thm]{Proposition}

\endgroup

\newcommand{\nwc}{\newcommand}
\nwc{\bit}{\begin{itemize}}
\nwc{\eit}{\end{itemize}}
\nwc{\Levy}{L\'evy}
\nwc{\LK}{L\'evy-Khintchine}
\nwc{\LI}{L\'evy-It\^{o}}
\nwc{\be}{\begin{equation}}
\nwc{\ee}{\end{equation}}
\nwc{\ba}{\begin{eqnarray}}
\nwc{\ea}{\end{eqnarray}}
\nwc{\la}{\label}
\nwc{\nn}{\nonumber}
\nwc{\Z}{\mathbb{Z}}
\nwc{\C}{\mathbb{C}}
\nwc{\E}{\mathbb{E}}
\nwc{\R}{\mathbb{R}}
\nwc{\Rplus}{\R_+}
\nwc{\N}{\mathbb{N}}
\nwc{\attr}{\mathcal{A}_{\rm p}}
\nwc{\attrx}{\mathcal{A}}
\nwc{\PP}{\mathcal{P}}
\nwc{\PPE}{\mathcal{P}(\Rplus)}
\nwc{\M}{\mathcal{M}}
\nwc{\Tt}{T^{(t)}}
\nwc{\Ut}{U^{(t)}}
\nwc{\gtx}{g^{(t)}_x}
\nwc{\fun}{f}
\nwc{\funinv}{f^{\dag}} 
\nwc{\law}{\stackrel{\mathcal{L}}{\rightarrow}}
\nwc{\eqd}{\stackrel{d}{=}}
\nwc{\vp}{\varphi}
\nwc{\Vp}{\varphi}
\nwc{\psilevy}{\Psi}
\nwc{\ve}{\varepsilon}
\nwc{\veps}{\varepsilon}
\nwc{\eps}{\ve}
\nwc{\betarsc}{\beta}
\nwc{\cl}{c\'{a}dl\'{a}g}
\nwc{\qref}[1]{(\ref{#1})}
\nwc{\D}{\partial}
\nwc{\Ebar}{{\bar{E}}}
\nwc{\ebar}{[0,\infty)}
\nwc{\mmt}{m}
\nwc{\fzero}{F_\rho} %{F_{0,\rho}}
\nwc{\fone}{M_\rho} %{F_{1,\rho}}
\nwc{\ip}[1]{\langle #1 \rangle}
\nwc{\ipbig}[1]{\left\langle #1 \right\rangle}
\nwc{\Lip}{\mathop{\rm Lip}\nolimits}
\nwc{\Tmin}{T_{\min}}
\nwc{\Tmax}{T_{\max}}
\nwc{\Tgel}{T_{\rm gel}}
\nwc{\LL}{\mathcal{L}}
\nwc{\mudot}{\mu}
\nwc{\nudot}{\nu}
\nwc{\rme}{{\rm e}}
\nwc{\rmi}{{\rm i}}
\nwc{\nsup}{^{(n)}}
\nwc{\thetasup}{^{(\theta)}}
\nwc{\tausup}{^{(\tau)}}
\nwc{\lambdasup}{^{(\lambda)}}
\nwc{\nsupj}{^{(n_j)}}
\nwc{\ksup}{^{(k)}}
\nwc{\jsup}{^{(j)}}
\nwc{\tsup}{^{(t)}}
\nwc{\nksup}{^{(n_k)}}
\nwc{\inv}{^{-1}}
\nwc{\iinv}{^{\dag}}
\nwc{\qxfac}{(1-\rme^{-qx})}
\nwc{\Mcan}{\mathcal{S}_d}
\nwc{\Mcanx}{\mathcal{S}}
\nwc{\sm}{G}
\nwc{\sv}{L}
\nwc{\rv}{R}
\nwc{\RV}{RV}
\nwc{\dsm}{H}
\nwc{\Hmap}{{\mathfrak S}_p}
\nwc{\Hmapx}{{\mathfrak S}}
\nwc{\canmap}{f}
\nwc{\canonical}{\mathcal{M}}
\nwc{\dnto}{\downarrow}
\nwc{\upto}{\uparrow}
\nwc{\cpair}{c} 
\nwc{\intR}{\int_0^\infty}
\nwc{\tnu}{\tilde\nu}
\nwc{\smeas}{g-measure} 
\nwc{\smeass}{g-measures} 
\nwc{\barsmeas}{$\overline{\mbox{g}}$-measure}
\nwc{\barsmeass}{$\overline{\mbox{g}}$-measures} 
\nwc{\nubar}{\bar{\nu}}
\nwc{\mubar}{\bar{\mu}}
\nwc{\dust}{a_0}
\nwc{\bdust}{b_0}
\nwc{\gel}{a_\infty}
\nwc{\bgel}{b_\infty}
\nwc{\ggel}{g_\infty}
\nwc{\Fbar}{\bar{F}}
\nwc{\Fcheck}{\check{F}}
\nwc{\nucheck}{\tilde{\nu}}
\nwc{\hvp}{\hat{\vp}}
\nwc{\hpsi}{\hat{\psi}}
\nwc{\htt}{\hat{t}}
\nwc{\hq}{\hat{q}}
\nwc{\hs}{\hat{s}}
\nwc{\dist}{\mathop{\rm dist}\nolimits}
\nwc{\extsol}{\Fbar}
\nwc{\specialF}{\hat\Fbar}
\nwc{\esm}{\bar{\sm}}
\nwc{\tlam}{\tilde\lambda}
\nwc{\lgr}{a} 
\nwc{\Fgr}{F_{\rho,\gamma}}
\nwc{\vpgr}{\vp_{\rho,\gamma}}
\nwc{\len}{l}
\nwc{\den}{\rho}
\nwc{\ldot}{\dot{\len}}
\nwc{\rhobar}{\bar{\rho}}
\nwc{\fbar}{\bar{f}}
\nwc{\one}{\mathbf{1}}
\nwc{\alphabar}{\bar{A}}
\nwc{\betabar}{\bar{\beta}}
\nwc{\rhoinit}{\hat\rho}
\nwc{\newrho}{F}
\nwc{\Rbar}{\bar{\newrho}}
\nwc{\Rbarhat}{\hat{\Rbar}}
\nwc{\lap}{\eta}
\nwc{\lapdsm}{\eta_{*}} %{\eta_\infty}
\nwc{\Q}{\mathbb{Q}}
\nwc{\forma}{\Vp}
\nwc{\formb}{\Psi}
\nwc{\invlen}{\len^{\dag}}
\nwc{\taumin}{\tau_*}
\nwc{\kappash}{\kappa^\#}

\theoremstyle{definition}
\newtheorem{defn}{Definition}[section]

\theoremstyle{remark}

\numberwithin{equation}{section}
\numberwithin{figure}{section}
\begin{document}
\title{Dynamics and self-similarity in min-driven clustering}
\author{Govind Menon\textsuperscript{1}, Barbara
  Niethammer\textsuperscript{2} and Robert L.  Pego\textsuperscript{3}}

\date{\today}

\maketitle

\begin{abstract}
We study a mean-field model for a clustering process that may be
described informally as follows. At each step a random integer $k$ is
chosen with probability $p_k$, and the smallest cluster merges with
$k$ randomly chosen clusters. We prove that the model determines a
continuous dynamical system on the space of probability measures 
supported in $(0,\infty)$,  
and we establish necessary and sufficient conditions for
approach to self-similar form. We also characterize eternal
solutions for this model via a \LK\/ formula. The analysis is based on 
an explicit solution formula discovered by Gallay and Mielke, extended
using a careful choice of time scale. 
\end{abstract} 


\footnotetext[1]
{Division of Applied Mathematics, Box F, Brown University, Providence, RI 02912.
Email: menon@dam.brown.edu}
\footnotetext[2]
{Mathematical Institute, University of Oxford,
Oxford, OX1 3LB, UK. Email: niethammer@maths.ox.ac.uk}
\footnotetext[3]{Department of Mathematical Sciences and 
Center for Nonlinear Analysis, 
Carnegie Mellon University, Pittsburgh, PA 15213.
Email: rpego@cmu.edu}

\pagebreak
%%%%%%intro
\section{Introduction}
\label{sec:intro}
\subsection{A mean-field model for clustering}
The clustering processes we consider are motivated by a simplified model
for domain wall motion in the one-dimensional Allen-Cahn equation 
$\D_tu=\D_{xx}u+u-u^3$. 
The domain walls become points on the line, and the domains are the intervals
separated by these points. The pattern coarsens by a simple rule:
At each step, the smallest domain combines with its two neighbors to
form a single domain, and this is repeated indefinitely.  
Computational simulations of this
`min-driven' domain coarsening process indicate that for a considerable
variety of initial distributions, the domain size distribution
approaches self-similar form \cite{CP,NK}. 

A mean-field model of this process was derived by Nagai and Kawasaki \cite{NK},
and it turns out to be amenable to a rigorous analysis aimed at explaining this
behavior \cite{CP,GM}.
We consider an infinite number of domains on the line, and
study the statistics of domain sizes using a number density
function $f(t,x)$. We assume that in any interval $I$ of unit length,
the expected number of domains with lengths in the range $(x,x+dx)$ is
given by $f(t,x)\, dx$.  The expected value of the total number of
domains in $I$ is denoted $N(t)= \int_0^\infty f(t,x) \, dx$. We
assume that $N(t)$ is finite, 
and denote the associated probability density by
\be
\label{defrho}
\rho_t(x) = \frac{f(t,x)}{N(t)}.
\ee
We let $\len(t)$ denote the size of the smallest domain at time $t$,
so $f(t,x)=0$ for $x<\len(t)$. 
The expected number of coalescence events per unit time is then
\be
\label{numevents}
f(t,\len) \dot{\len}.
\ee
The coalescence events affecting domains of size $x$ in the
time interval $(t,t+dt)$ are
(a) loss: consecutive domains of size $x,\len,y$ (or $y,\len,x$)
combine to form a domain of size $x+\len+y$; (b) gain: domains of size
$y,\len,x-y-\len$ combine to form a domain of size $x$. Under the
mean-field assumption that coalescing domains have sizes
chosen randomly and independently from the current overall size
distribution, these events have respective relative probability density
\[ \rho_t(x)  \rho_t(y) , \quad \rho_t(y)  \rho_t(x) , \quad \rho_t(y)
\rho_t(x-y-\len).\]
The rate equation for the evolution of $f$ is obtained by summing over
all loss and gain terms:
\ba
\label{evol1}
\lefteqn{\partial_t f(t,x) =}\\
\nn
&& f(t,\len) \dot{\len} \left( \int_{\len}^{x-\len} \rho_t(y)
  \rho_t(x-y-\len)\, dy - 2 
  \rho_t(x) \int_{\len}^\infty \rho_t(y)\, dy \right), \quad x >\len.
\ea

Clustering phenomena are seen in fields as varied as population
genetics and physical chemistry. Thus, while coarsening of intervals
provides concrete motivation, the notion of a `cluster' can have 
widely different interpretations in applications.  
The model \qref{evol1} above is one of a family of 
what we call {\em min-driven clustering} models 
that can be analyzed together in one setting as in \cite{GM}.
At each step a random integer $k \geq 1$ is chosen with
probability $p_k$, and the smallest cluster merges with $k$ randomly
chosen clusters. The mean-field assumption is that all these random
variables are independent. The only assumptions we impose on the probabilities
are that
\be
\label{prob}
p_k \geq 0, \quad \sum_{k=1}^\infty p_k =1, \quad \sum_{k=1}^\infty
k p_k < \infty.
\ee
The evolution of the number density under this process is described by
the following rate equation for  cluster size density, obtained 
by summing over all loss and gain events as in \qref{evol1}: 
\be
\label{evol2a}
\partial_t f(t,x) = f(t,l) \dot{l} \sum_{k=1}^\infty p_k \left(
  \rho_t^{\star k} (x-l) - k\rho_t(x)\right), \quad x > \len(t).
\ee
Here the notation $\rho_t^{\star k}$ denotes $k$-fold
self-convolution. The case $p_1=1$ corresponds to a ``paste-all'' model
discussed by Derrida et al.\ \cite{Derrida}. Also see \cite{IKR} for a
model of social conflict with rather similar solution formulas.

An important feature of the model \qref{evol2a} is its invariance under
reparame\-trization in time. If we change variables via $t=
T(\tilde{t})$, $\tilde{f}(x,\tilde{t})=f(x,t)$,
$\tilde{\len}({\tilde{t}})=\len(t)$, then equation \qref{moment1}
retains its form since 
\[ 
\partial_{\tilde{t}}\tilde{f}(x,\tilde{t}) =
\dot{T} \partial_t f, \quad \partial_{\tilde{t}}\tilde{\len} 
= \dot{T} \partial_{t}\len. 
\] 
A careful choice of the time scale is key to the analysis.  In the
first mathematical study of \qref{evol1}~\cite{CP}, the authors imposed the
relation $f(t,\len) \dot{\len}=1$, meaning the number of coalescence
events per unit time is constant~\cite{CP}.  In a more recent paper,
Gallay and Mielke parametrized time by the minimum size, that is
$\len(t)=t$~\cite{GM}. This leads to an elegant solution procedure
that was used to prove some basic results on well-posedness and 
the approach to self-similarity. Gallay and Mielke showed that
\qref{evol2a} defines a strongly continuous flow in $L^1$, that
\qref{evol2a} admits a one-parameter family of self-similar solutions
and that suitable initial densities yield convergence to self-similar
form.  Precise comparisons between these results and ours are made
later in this paper.

In this article, we introduce yet another time scale. We
parametrize time inversely to the total number of domains, so that
\be
\label{intime}
t= N(t)^{-1}.
\ee
We shall argue that this is a natural choice for a number of
reasons. It retains the simplicity of the choice $\len(t)=t$, and 
allows us to obtain: (a) existence and uniqueness for measure-valued solutions; 
(b) necessary and sufficient conditions for convergence to self-similar 
form; (c) a characterization of eternal solutions for
 the dynamical system defined by \qref{evol1}. We comment on these in
 greater depth below. 

Let us first explain one simple motivation for  \qref{intime}. 
We show below that for \qref{evol1}, $f(t,\len)\dot{\len}=-\dot{N}/2$. 
Thus, $f(t,\len)\dot{\len}=N^2/2$ when \qref{intime} 
holds, and \qref{evol1} now takes the form
\ba
\label{evolsmol}
\lefteqn{\partial_t f(t,x) =}\\
\nn
&& \frac{1}{2}\int_{\len}^{x-\len} f(t,y) f(t,x-y-\len)\, dy -  
  f(t,x) \int_\len^\infty f(t,y)\, dy, \quad x >\len.
\ea
If we take $\len(t)\equiv 0$ (as a model when the smallest domains
have negligible size, for example), this reduces to a basic solvable model of
clustering: Smoluchowski's coagulation equation with constant kernel.
In recent work, we provided a comprehensive analysis of dynamic
scaling in this equation by exploiting an analogy with the classical
limit theorems of probability theory~\cite{MP1,MP3}. We use these insights
to guide our study of \qref{evol2a}.

\subsection{Measure-valued solutions}
Mean-field models of domain coarsening, such as the LSW model, Smoluchowski's
coagulation equation, and \qref{evol2a}, correspond to physical
processes where mass is transported from small to large scales. 
For several reasons, it is natural to consider measure-valued solutions,
for which the size distribution need not have a continuous or
integrable density.  Such solutions are physically meaningful, as many
clustering processes (e.g., polymerization) involve
a discrete set of sizes based on an elementary unit. 
A formulation via measures is also mathematically elegant, as it allows us
to unify the treatment of discrete and continuous coagulation
models, exploit simple criteria for compactness and continuity, and
prove basic uniform estimates.

To any solution of \qref{evol2a} we associate a probability
measure $F_t$ with distribution function written
\begin{equation}
F_t(x) = \frac1{N(t)}\int_0^x f(t,y)\,dy  = \int_0^x \rho_t(y)\,dy.
\end{equation}
For any probability distribution $F$ on $[0,\infty)$, we call
$\len=\inf\{x| F(x)>0\} $ the {\em min} of $F$. (This is short for ``minimum
size,'' regarding $F$ as a probability distribution for size.)
We will prove that the initial-value problem for an appropriate weak
form of \qref{evol2a} is well-posed for initial probability measures
$F_{t_0}$ with positive min.  That is, \qref{evol2a} with
\qref{intime} determines a
continuous dynamical system on the space of probability measures with
positive min, equipped with the weak topology. 
See Theorem~\ref{thm.wp} below.   By comparison, Gallay and Mielke
established that the initial-value for \qref{evol2a} defines a
continuous dynamical system on the space of probability {\em
densities\/} in $L^1(1,\infty)$ equipped with the strong
topology~\cite[Thm. 3.3]{GM}. 
The solutions we construct arise by a natural completion of
these $L^1$ dynamics.

\subsection{Dynamic scaling}
A common theme in recent studies of dynamic scaling in
mean-field models of 
coarsening is that the approach to self-similarity is both degenerate
and delicate. The problem is  degenerate because there is a one-parameter
family of self-similar solutions. For the model studied here with
$p_k=0$ for all $k$ large enough, Gallay
and Mielke found a family of self-similar solutions that may be 
rewritten in the time scale \qref{intime} in the form
\be
\label{ftheta1}
F_t(x) = F\thetasup\left(\frac{x}{l\thetasup(t)}\right), \quad l\thetasup(t) =
t^{1/\theta}, \quad \theta \in (0,1], \quad t >0.
\ee
Here $F\thetasup$ is a  probability distribution with density
$\rho\thetasup$ supported on $[1,\infty)$. 
The density $\rho\thetasup$ is known explicitly only through its
Laplace transform.  Only $\rho^{(1)}$ has finite
mass (first moment); $\rho^{(1)}(x)$ decays exponentially as
$x\to\infty$. 
The distributions $F\thetasup$ for $0<\theta<1$
have heavy tails, with  
$\rho\thetasup(x) \sim c_\theta x^{-(1+\theta)}$ as $x\to\infty$
(see Theorem~\ref{t.ss}).

The problem is delicate because the domains of
attraction of the self-similar solutions are determined by the tails
of the initial size distribution, in the precise manner explained below.
(See~\cite{MP1,NP1} for analogous results on the LSW model of Ostwald
ripening and Smoluchowski's coagulation equations with solvable kernels.) 
Gallay and Mielke showed that all densities with
finite mass are attracted to the self-similar solution with
$\theta=1$. Moreover, for $0<\theta \leq 1$ they showed that if the
initial data $\rho_{t_0}$ is sufficiently close to $\rho\thetasup$ in
a suitable weighted norm then the rescaled probability density $t\rho_{t}(tx)$
approaches $\rho\thetasup$ with a rate of convergence determined by
the weighted norm (see~\cite[Thms. 5.5, 5.7]{GM}). These results provide
sufficient conditions for approach to self-similarity. 

Our aim is to establish conditions that are both necessary and
sufficient to answer a more general question about
{\em arbitrary\/} scaling limits.
We characterize the set of all non-degenerate limits under a
general rescaling  of the form $F_t(\lambda(t) x)$ where 
$\lambda(t)$ is a measurable, positive function such that $\lim_{t \to \infty}
\lambda(t)=\infty$. A limit is non-degenerate
if it is suitable data for the initial-value problem. That is,
non-degenerate limits are probability distributions with a positive min.
\begin{thm}
\label{thm.domains}
Let $t_0>0$, let $F_{t_0}$ be an arbitrary probability measure on
$(0,\infty)$ with positive min,
and let $F_t$ ($t\ge t_0$) 
be the associated measure-valued solution of \qref{evol2a}
(see Theorem~\ref{thm.wp}).
\begin{enumerate}
\item[(i)] Suppose there is a measurable, positive function
  $\lambda(t)\to\infty$ as $t\to\infty$ 
and a probability measure $F_*$ with positive min, such that 
\be
\label{conv1}
\lim_{t \to \infty} F_t(\lambda(t)x) = F_*(x)
\ee
at all points of continuity of $F_*$. Then there exists $\theta \in
(0,1]$ and a function $L$ slowly varying at infinity such
that the initial data $F_{t_0}$ satisfies
\be
\label{conv2}
\int_0^x y F_{t_0}(dy) \sim x^{1-\theta} L(x), \quad
\mathrm{as}\quad x \to \infty.
\ee
Moreover, the min $\len(t)$ of $F_t$ and the rescaling $\lambda(t)$
satisfy 
\be
\label{conv3}
 \lambda(t)l_* \sim \len(t) \sim t^{1/\theta} \tilde{L}(t),\quad
\mathrm{as}\quad t \to \infty,
\ee
where $l_*>0$ is the min of $F_*$
and $\tilde{L}$ is slowly varying at infinity, related to $L$ by
\qref{e.len}.

\item[(ii)] Conversely, assume there exists $\theta \in (0,1]$ and a
  function $L$ slowly varying at infinity such that the initial data
  satisfies \qref{conv2}. Then $\len(t)$ satisfies
  \qref{conv3}, and 
\be
\label{conv4}
\lim_{t \to \infty} F_t(\len(t) x) = F\thetasup(x), \quad x \in (0,\infty).
\ee
\end{enumerate}
\end{thm}

A positive function $L$ is slowly varying at infinity if it is
asymptotically flat under rescaling in the sense that $\lim_{x \to
  \infty}L(xy)/L(x)=1$ for every $y>0$. For example, all powers and
iterates of the logarithm are slowly varying at infinity. These are
the class of admissible corrections to the power law 
$x^{1-\theta}$. 

Part (i) of the theorem is an assertion of rigidity of
scaling limits. We assume only that $\lambda(t)$ is measurable, positive and
$\lim_{t \to \infty} \lambda(t) =\infty$. It then follows that the
limits must define self-similar solutions, and $\lambda(t)$ must be
the time scale associated to the self-similar solution, up to a slowly
varying correction. Part (ii) of the theorem, and the sufficient
conditions of Gallay and Mielke, show that the domains of attraction
are determined by the tails of the initial data. As a consequence of
part (i) of the theorem, the condition \qref{conv2} is optimal.

Gallay and Mielke \cite{GM} used a rather delicate Fourier analysis 
to establish the existence for self-similar solutions by studying
their densities.
We will use the proof of part (ii) of the theorem above to simplify 
much of this analysis and extend it to the case when $p_k\ne0$ for
infinitely many $k$. In this we are motivated by a certain resemblance
of the min-driven model to hydrodynamic limits of
what are called {$\Lambda$-coalescents} in probability theory
\cite{BlG}, which are  
clustering processes involving arbitrarily many multiple collisions.
We find a curious fact, namely,  when $\sum p_k k\log k=\infty$ there
is {\em no} self-similar solution with finite mass (first moment). 
Still, the theorem above correctly describes the domains of
attraction. Solutions with finite mass approach the self-similar
solution with $\theta=1$, but this self-similar solution has infinite
mass.

\subsection{Eternal solutions}
Theorem~\ref{thm.domains} is a particular example of the principle
that the asymptotic behavior under rescaling is determined by the tail of the
initial distribution. The dynamics exhibit sensitive dependence
on initial conditions, as arbitrarily small changes in the tail of the
initial data can lead to widely divergent asymptotic behavior. This
indicates a kind of chaos, and it is of interest to find a precise
formulation. A comprehensive analysis of such phenomena for the
solvable cases of Smoluchowski's coagulation equation appears
in~\cite{MP3}. This analysis is guided by an analogy with the
probabilistic notion of {\em infinite divisibility\/}. Let us first
describe these results informally. 

Clustering is an irreversible process, and in general we do not expect
to be able to solve \qref{evol1} backwards in time (for $t<t_0$). However, the
self-similar solutions have the remarkable feature that they are
defined for all $t >0$. That is, they are {\em divisible\/} under
the coalescence process. We call a solution {\em eternal\/} if it is
defined on the maximal interval $(0,\infty)$ consistent with
\qref{intime}. In probability theory,
the infinitely divisible distributions are characterized by the
celebrated \LK\/ formula.  In~\cite{MP3} we extended a result of
Bertoin~\cite{B_eternal} showing that the class of eternal solutions
to Smoluchowski's coagulation equations is also characterized by a
\LK\/ formula.  Heuristically, this formula describes the emergence of
eternal solutions from infinitesimally small clusters at $t=0$.
We also showed that the set of all subsequential limits---the
scaling attractor---is in a one to one correspondence with the
eternal solutions. A rigorous description of chaos is based on the
fact that nonlinear dynamics on the scaling attractor is reduced to 
linear scaling dynamics using the \LK\/ formula.

In this article, we take the first step towards establishing a similar
picture for min-driven clustering.  Namely, we prove  a \LK\/
formula characterizing all eternal solutions for min driven clustering
(Theorem~\ref{thm:LK}).  The choice of time scale
\qref{intime} is very convenient for this analysis.

\subsection{Outline}
The rest of the article is organized as follows. We describe  
the solution procedure of Gallay and Mielke in the next section and 
discuss how the number-driven time scale is motivated by
the important example of initial data that are monodisperse (a Dirac delta).
The treatment here is formal. We establish
some analytic prerequisites in Section~\ref{sec:prelim}. This is
followed by  rigorous results: the proof of well-posedness for
measure-valued solutions (see Theorem~\ref{thm.wp}) in Section~\ref{sec.ds}, 
the study of self-similar profiles in Section~\ref{sec.ss}, the
characterization of domains of attraction in Section~\ref{sec.rv},  
and the characterization of eternal solutions in Section~\ref{sec:lk}. 
%%%%%%%%%%%%%%%%%%%end intro%%%%%%%%%%%%%%%%%%%%
%%%%%%%
%%%%%%% input sol %%%%%%%%%%%%%%%%%%%%%%%%%%%
\section{The solution formula for min-driven clustering}
\label{sec.sol}
\subsection{The generating function and moment identities}
As in the theory of branching processes, it is convenient to keep 
track of the clustering process with a generating function  
\be
\label{genfn}
Q(z) = \sum_{k=1}^\infty p_k z^k. 
\ee
For example, binary clustering as in  the Allen-Cahn model corresponds
to $Q(z)=z^2$. The  generating function $Q$ is analytic in the unit
disk $\{|z| <1\}$,  and absolutely monotone (that is, $Q$ and all its
derivatives are positive on $[0,1)$).
We assume that the expected number of clusters in the mergers, denoted $Q_1$,  
is finite. That is, 
\be
\label{genfn2}
Q_1= Q'(1)= \sum_{k=1}^\infty k p_k < \infty.
\ee
We define a convolution operator $\Q(\rho) = \sum_{k=1}^\infty p_k 
\rho_t^{\star k}$ associated to $Q$, and rewrite
\qref{evol2a} in the form 
\be
\label{evol2}
\partial_t f(t,x)= f(t,l) \dot{l}\left( \Q(\rho_t)(x-\len) - Q_1 \rho_t(x)
\right), \quad x > \len(t).
\ee

We extend the evolution equation \qref{evol2} from densities to
measures as follows. We consider a number measure $\nu_t$ and a probability
measure $F_t$ that are related to the densities (when they exist)
by  
\begin{equation}\label{d.nuF}
\nu_t(dx) =f(t,x) \, dx, \qquad 
F_t(dx) = \frac{\nu_t(dx)}{N(t)} = \rho_t(x)\,dx .
\end{equation}
Let $\Rplus$ denote the interval $[0,\infty)$. If $a\colon \Rplus \to
\C$ is continuous with compact support, then formally
\[ \frac{d}{dt} \int_{\Rplus} a(x) f(t,x) \, dx  = \int_{\len}^\infty
a(x) \partial_t f(t,x) 
dx - a(\len)f(t,\len) \ldot. \]
We substitute for $\partial_t f(t,x)$ using \qref{evol2}, 
\qref{d.nuF} to obtain the moment identity
\ba
\label{moment1}
\lefteqn{ \frac{d}{dt} \int_{\Rplus} a(x) \nu_t(dx) \,= }\\
\nn
&& f(t,\len) \ldot \sum_{k \geq 1}
p_k \int_{\Rplus^k}\left[ a\left(\len + \sum_{i=1}^k y_i \right)- a(\len)
 -\sum_{i=1}^k a(y_i)  \right] \prod_{i=1}^k F_t(dy_i)\, . 
\ea
Some basic properties of the model are obtained by choosing suitable
test functions $a$ in \qref{moment1}. We set $a(x)=x$ to see that
mass is conserved:
\be
\label{mass}
\frac{d}{dt} \int_0^\infty x \, \nu_t(dx) = 0.
\ee
When  $a=1$, we obtain the rate of change of the total number of
clusters,
\be
\label{number}
\dot{N} = - Q_1 f(t,l) \ldot = -N Q_1 \rho_t(l) \ldot.
\ee
We substitute \qref{number} in \qref{evol2} to see that 
$\rho_t$ satisfies
\be
\label{evol3}
\partial_t \rho_t = \rho_t(\len) \ldot \, \Q(\rho_t)(x-l), \quad x >l.
\ee
Similarly, we use
\qref{genfn2}, \qref{moment1} and  \qref{number} to obtain the moment identity
\be
\label{moment2}
 \frac{d}{dt} \int_{\Rplus} a(x) F_t(dx) = \rho_t(\len)\ldot
 \sum_{k=1}^\infty p_k \int_{\Rplus^k} \left[ a\left(\len +
     \sum_{i=1}^k y_i \right) -a(\len)\right] \prod_{i=1}^k
 F_t(dy_i)\, .
\ee
This identity becomes an appropriate weak form of \qref{evolsmol}
after we later impose our choice of time scale $t=1/N$, which means from
\qref{number} that 
\begin{equation}
\rho_t(\len)\ldot = \frac1{Q_1t}.
\label{rholdot}
\end{equation}

\subsection{Gallay and Mielke's solution formula}
\label{subsec:GM}
A remarkable feature of these min-driven clustering models is that the
evolution equation admits an elegant solution via the Fourier (or
Laplace) transform. Our analysis relies heavily on this solution
procedure, due to Gallay and Mielke \cite{GM}.
The main difference with \cite{GM} is that we prefer to use the Laplace
transform, denoted by 
\be
\label{LTdef}
\rhobar_t(q) = \int_{\Rplus} e^{-qx} \rho_t(x) \, dx, \quad q>0. 
\ee
We set $a(x)=e^{-qx}$ in \qref{moment2}  to obtain the ordinary
differential equation
\be
\label{LT}
\partial_t \rhobar_t(q) = -(\rho_t(l)\ldot) \,  e^{-ql}
\left(1-Q(\rhobar_t(q))\right). 
\ee
In order to integrate this equation, we define an analytic function
$\forma$ via 
\be
\label{defPhi}
\forma'(z) = \frac{Q_1}{1-Q(z)}, \quad \forma(0)=0.
\ee
(This definition of $\forma$ differs by the factor $Q_1$ from that used
in~\cite{GM}.) $\forma$ is strictly increasing on $[0,1)$. 
In the case of binary clustering, $Q(z)=z^2$, and the functions $\forma$
and $\forma^{-1}$ are  
\be
\label{binary}
 \forma(z) = \log \left(\frac{1+z}{1-z} \right), \quad
\forma^{-1}(w) = \tanh \frac{w}{2}.
\ee

We substitute \qref{defPhi} in  \qref{LT} to obtain
\be
\label{sol2}
\partial_t \forma(\rhobar_t(q)) = - (Q_1 \rho_t(l) \ldot)\, e^{-ql},
\quad q>0.
\ee
The choice of time scale has played no role in the analysis this
far. Gallay and Mielke  parametrize time by the minimum cluster size,
setting $\ldot =1$. For clarity of notation, we denote this
choice of time scale by $\tau$, reverting to the letter $t$ when we
introduce the number-driven time scale in \qref{intime}. 

With $\len(\tau)=\tau$, the value of $\rho_\tau$ on the free boundary
$x=\tau$ plays  an important role in the solution. 
We use this density to define a measure on $(0,\infty)$ that we call the 
{\em trace measure $A$\/}, with distribution function written
\be
\label{alphadef}
A(\tau) = 
\begin{cases}
\alpha_0+\int_{\tau_0}^\tau Q_1 \rho_s(s) \, ds ,
&\tau\ge\tau_0,\cr
\alpha_0 & \tau<\tau_0.
\end{cases}
\ee 
Here $\tau_0$ denotes the initial time, and $\alpha_0$ is
any convenient constant.
Equation \qref{sol2} may now be rewritten  
\be
\label{sol1}
\partial_\tau 
\forma(\rhobar_\tau(q))
= -e^{-q\tau} \frac{dA}{d\tau}. 
\ee
Fix $\tau_1> \tau_0$.  We integrate \qref{sol1} from $\tau_0$ to
$\tau_1$ to obtain  
\be
\label{sol3}
\forma(\rhobar_{\tau_1}(q)) - \forma(\rhobar_{\tau_0}(q)) 
= - \int_{\tau_0}^{\tau_1} e^{-qs} A(ds). 
\ee
Since $ \rho_\tau$ is supported in $[\tau,\infty)$, we have the
estimate
\[ \rhobar_\tau(q) \leq e^{-q\tau} \int_\tau^\infty \rho_\tau(x) \, dx
= e^{-q\tau}. \]
We now let $\tau_1 \to \infty$ in \qref{sol3} to find the Laplace
transform of $A$ given by 
\be 
\label{sol4}
\alphabar(q)= \int_{\Rplus} e^{-qs}A(ds) = \forma(\rhobar_{\tau_0}).
\ee
Thus, the trace measure $A$ is the inverse Laplace transform of
$\forma(\rhobar_{\tau_0})$ and is determined completely by the initial
data.  

We may now repeat this argument to determine the solution at any time
$\tau > \tau_0$. We replace $\tau_0$ by $\tau$ in \qref{sol3}, let
$\tau_1 \to \infty$, and obtain 
\be
\label{sol5}
\rhobar_\tau(q) = \forma^{-1} \left(\alphabar_\tau(q) \right), 
\ee
where 
\begin{equation}
\label{sol6}
\alphabar_\tau(q) = \int_\tau^\infty e^{-qs} A(ds). 
\end{equation}
We note that $\alphabar_\tau$ is the Laplace transform of the truncated trace
measure  $A_\tau$ satisfying
\be\label{gl1}
A_\tau(ds) = H(s-\tau) A(ds),
\ee
where $H$ is the Heaviside function. 
We may write
\be
\label{gl}
A_\tau(s) = A(s)\vee A(\tau)
\ee
for the distribution function, with the notation $a\vee b=\max(a,b)$.
Therefore, the {\em nonlinear\/}
evolution of $\rho_\tau$ is determined by the {\em linear\/} evolution 
of $A_\tau$. This global linearization underlies the analysis
in~\cite{GM}. 

\subsection{The number-driven time scale and an extended solution formula}
It is natural to try to use the formula \qref{sol5} as a
basis for finding measure-valued solutions when the initial data
$\rho_{\tau_0}$ is replaced by an arbitrary probability distribution.
However, we face two difficulties. The first is that it is not clear 
that \qref{sol5} necessarily {\em defines\/} a measure $\rho_\tau$. 
That is, it is not clear that the right-hand side of \qref{sol5}
is necessarily the Laplace transform of a measure.
The second difficulty is that 
whenever the trace measure $A$ has atoms, 
any solution defined through \qref{sol5} is
discontinuous in time, as is clear from \qref{sol6}.
We overcome the first difficulty by an approximation
argument. This relies on the simple and fundamental fact that a limit
of completely monotone functions is completely monotone. 
We overcome the second difficulty by switching to the number-driven time
scale \qref{intime}. 

Henceforth, the letter $t$ always denotes the number-driven
time scale. The measure-valued solution is
denoted by $\newrho_t$, its Laplace transform by 
\[
\Rbar_t(q) = \int_0^\infty \rme^{-qr}\,\newrho_t(dr),
\]
and the minimum cluster size by $\len(t)$. These are related to
solutions in the time scale $\tau$ by 
\be
\label{newold}
\len(t)=\tau,  \quad\newrho_t(dx)= \rho_\tau(x) \, dx, \quad
\Rbar_t(q)=\rhobar_\tau(q). 
\ee 
We now rewrite the solution formula \qref{sol5} in terms of measures. 
The relation $\len(t)=\tau$,  equation \qref{number} and the
definition of $A$ in \qref{alphadef} imply
\be
\label{dtdtau}
\frac{dt}{t} = Q_1 \rho_\tau(\tau) d\tau = A(d\tau). 
\ee
The differential equation \qref{LT} now takes the form
\be
\label{newsol}
\partial_t \Rbar_t(q) = -\frac{e^{-q\len(t)}}{Q_1t}(1-Q(\Rbar_t(q)).
\ee
If $t_0$ denotes the initial time, we may integrate equation
\qref{dtdtau} to obtain 
\be
\label{defL}
\log \left( \frac{t}{t_0} \right)= A(l(t))  -\alpha_0.
\ee
The change of variables \qref{newold} and \qref{dtdtau} affects the
Laplace transform of $A_\tau$ as follows:
\be
\label{newsol3}
\int_{\tau}^\infty e^{-qr} A(dr) = \int_t^\infty e^{-q\len(s)} \frac{ds}{s}. 
\ee

These calculations yield the following  revised solution procedure:
Given an arbitrary initial probability measure $F_{t_0}$,
the trace measure $A$ is found as in \qref{sol4}
by inverting its Laplace transform, given by 
\be \label{newsol4}
\alphabar(q) = \forma(\Rbar_{t_0}(q)). 
\ee
Next, we determine $\len(t)$ through inverting
\qref{defL}. Once $\len(t)$ is known, 
the solution $\newrho_t$ is determined by inverting 
the Laplace transform given as in \qref{sol5} and \qref{sol6} by 
\be
\label{newsol5}
\Rbar_t(q) = \forma^{-1} \left( \int_t^\infty e^{-q\len(s)}
  \frac{ds}{s}\right), \quad t \geq t_0. 
\ee

The main observation is that working with $\len(t)$ instead of
$A_\tau$ yields an evolution continuous in time.  
Since $\len(t)$ is an increasing function, it has at worst jump 
discontinuities. But \qref{newsol5} shows  that discontinuities in
$\len$ do not affect the continuity in $t$ of $\Rbar_t(q)$, and thus
the continuity of $\newrho_t$ in the weak topology.

\subsection{An example: monodisperse initial data}
Let us illustrate the meaning of the extended solution formula 
in the new time scale with an important example. 
Set $t_0=\tau_0=1$, $Q(z)=z^2$ and $F_1(x) =
\mathbf{1}_{x\geq1}$. That is, initially all clusters have size 1.
Then $\Rbar_1(q)  
=e^{-q}$, and 
\[
\alphabar(q) = \forma(\Rbar_1(q))=  \log(1+ e^{-q}) -
  \log(1-e^{-q}) 
.\]
We differentiate with respect to $q$ and simplify to obtain
\[
-\partial_q \alphabar(q) = \frac{2e^{-q}}{1-e^{-2q}} =  2\left( e^{-q} +
e^{-3q} + e^{-5q} + \ldots.\right)  \]
Since $e^{-kq}$ is the Laplace transform of $\delta_k(dx)$, the
trace is
\be
\label{trace-mono}
A(x) = \sum_{k \leq x, \;k \;\mathrm{odd}} \frac{2}{k}.  
\ee
$A$ has jump discontinuities at the odd integers. We shall work with
the right continuous inverse, so that the minimum cluster size is 
\be
\label{monosol1}
\len(t) = k, \quad t \in [t_{k-2},t_k),  \quad t_k = e^{2\left( 1 +
  \frac{1}{3} + \ldots \frac{1}{k} \right)}, \quad k \;\;\mathrm{ odd},
\ee
with $t_{-1}=t_0=1$. The solution formula \qref{newsol5} now yields
\be
\nn
\forma(\Rbar_t(q)) = \left(\log \frac{t_k}{t}\right) \frac{e^{-kq}}{2} +
\frac{e^{-q(k+2)}}{k+2} + \frac{e^{-q(k+4)}}{k+4} + \ldots, \quad
t \in [t_{k-2}, t_{k}). 
\ee
The solution  has the following interpretation. For $t \in
[t_{k-2},t_k)$, $F_t$ is supported  on the odd integers greater than
or equal to $k$.  The fraction of 
clusters of size $k$ decays continuously to zero over the time interval
$[t_{k-2},t_k)$.  Thus, the number-driven time scale regularizes jumps in
the trace measure by providing a finite time for these jumps to
vanish. A moment's reflection suggests that this is what we should
expect if we approximate the monodisperse data by a smooth density.
%%%%%%%%%%%%%%%%%%%end sol%%%%%%%%%%%%%%%%%%%%
%%%%%%%
%%%%%%% input lt %%%%%%%%%%%%%%%%%%%%%%%%%%%
\section{Analytic preliminaries}
\label{sec:prelim}
This section is a summary of the main analytic methods we use. We
present some facts about distribution functions, Laplace transforms
and Tauberian theorems in a form suitable for use in later sections.

\subsection{Distribution functions}
We will consider measures on an interval $J \subset \R$. We study a
measure through its distribution function, often using the same
notation for both.  It is therefore convenient to introduce 
the following conventions for brevity. 
A distribution function $\fun\colon J \to \R$  is 
a right-continuous, increasing function. 
(``Increasing'' means that $x_1\le x_2$ implies $\fun(x_1)\le \fun(x_2)$.)
A distribution function $\fun$ is identified with a
measure via
\be
\label{measuredist}
\fun\left((x_1,x_2] \right) =\fun(x_2) -\fun(x_1).
\ee
We do not assume that the function $\fun$ is positive, 
since the trace measure for
self-similar solutions is of the form $A(\tau) = \theta \log \tau$,
$\tau >0$.  Following probabilistic convention, we say
a sequence $\fun_n$ of measures on $J$ converges weakly to $\fun$ 
(written $\fun_n\to \fun$) if and only if 
$\fun_n([a,b])\to \fun([a,b])$ as $n\to\infty$ whenever $a,b\in J$
are not atoms of $\fun$, meaning $\fun(\{a\})=\fun(\{b\})=0$.
We have $\fun_n\to \fun$ if and only if
at every point of continuity of $\fun(x)$,
\[
\fun_n(x) + c_n \to \fun(x) \quad\mbox{as $n\to\infty$},
\]
for some constants $c_n$ independent of $x$.

Given a distribution function $f\colon J\to\R$, its epigraph is 
the set
\[
\Gamma(f) = \{(x,y)\in\R^2\mid f(x^-)\le y\le f(x), x\in J\}.
\]
There is a unique distribution function $\funinv$, 
with epigraph obtained by reflection through $x=y$:
$\Gamma(\funinv)=\{(x,y)\mid(y,x)\in\Gamma(\fun)\}.$
We call $\funinv$ the {\em inverse} of $\fun$.
We can write
\be
\label{finv}
\funinv(\tau) = \inf\left\{ t \in J \left| \fun(t) > \tau
 \right. \right\}, 
\quad \tau<\sup\{ \fun(t)\mid t\in J\}.
\ee
The following convergence result is not difficult to prove.
\begin{lemma}
\label{inverseA}
Suppose $\fun_n\colon J \to \R$ is a sequence of distribution functions that
converges to $\fun$ at all points of continuity. Then 
for every point of continuity $x$ of $\funinv$, 
$\funinv_n(x)$ is defined for sufficiently large $n$, 
and $\lim_{n \to \infty} \funinv_n(x) = \funinv(x)$.
\end{lemma}

\begin{lemma}
\label{approxincr}
Suppose $\fun \colon J \to \R$ is a distribution
function. There exist monotonically increasing and monotonically
decreasing  sequences of piecewise constant distribution functions
that converge to $f$ at all points of continuity.
\end{lemma}
\begin{proof}
We assume that $J=[0,\infty)$ for clarity, and we only
construct an increasing sequence of approximations. The argument
is easily generalized. Suppose 
$k \geq 1$ is an integer. Let $\one_k$ denote the indicator function
$\one_k(x) = 1$, $k \leq x < k+1$, and $\one_k(x)=0$ otherwise. 
Consider the sequence of functions
\[
   g_n(x) = \sum_{k=0}^\infty f\left(\frac{k}{n}\right)\one_k(nx). 
\]
Since $f$ is increasing, we have 
\[ g_1(x) \leq g_2(x) \leq \ldots \leq g_n(x) \leq f(x). \]

Let us establish convergence of $g_n$.  Given $x \geq 0$
and $n$, let $k_n$ denote the  integer such 
that $k_n \leq nx < k_n+1$. Clearly $\lim_{n \to \infty}
k_n/n=x$. Suppose $x$ is a point of continuity of $f$. Then 
$g_n(x) = f(k_n/n)$, therefore $g_n(x) \to f(x)$. 
\end{proof}

\subsection{Min history and trace}
Our analysis will focus on two related distribution functions---the 
minimum cluster-size history $\len$ and the trace $A$. 

Fix $t_0>0$. A {\em min history} is a 
positive distribution function $\len$ on $[t_0,\infty)$. 
In this article we also require a min history to be unbounded:
$\len(t)\to\infty$ as $t\to\infty$.  Given any min history $\len$,
we associate a {trace} $A$ on $(0,\infty)$ via
\be
\label{defA2}
A(\tau) = \log \invlen (\tau), 
\quad \tau >0.
\ee
Note $A(\tau)=\log t_0$ for $0<\tau <\len(t_0)$ due to \qref{finv}.

Conversely, we say $A$ is a {\em trace} if it is
a distribution function on $(0,\infty)$ such that 
(i) $A(\tau)=\log t_0$ on some nonempty, maximal interval
$(0,\tau_0)$, and (ii) $A(\tau)\to\infty$ as $\tau\to\infty$.  
Given any trace $A$, we can associate a min history $\len$ by 
\be
\label{defL2}
\len(t) = \exp A\iinv(t), 
\quad t \ge t_0.
\ee

\begin{prop}[Change of variables]
\label{propCV}
(a) Assume $\len$ is a min history and there exists $c \geq 0$
such that
\be
\label{lfinite}
\int_{t_0}^\infty e^{-q\len(t)} \frac{dt}{t} < \infty, \quad q\in (c,\infty).
\ee
Then the trace $A$ given by \qref{defA2} satisfies
\be
\label{lAfinite}
\alphabar(q)= \int_0^\infty e^{-q\tau} A(d\tau) =  \int_{t_0}^\infty
e^{-q\len(t)} 
\frac{dt}{t}, \quad q \in (c,\infty).
\ee
(b) Assume $A$ is a trace such that $\alphabar(q) < \infty$
for $q \in (c,\infty)$, and let $\len(t)$ be defined by
\qref{defL2}. Then \qref{lAfinite} holds. 
\end{prop}
\begin{proof}
 {\em 1.\/}  We first verify the equality for piecewise constant functions.
Suppose $0< t_0 < t_1 < t_2 < \ldots$ and $0< l_0 < l_1 <
\ldots$ are strictly increasing sequences. Consider the piecewise
constant, increasing function
\be
\label{lPL}
\len(t) = \sum_{k=0}^\infty \len_k \one_{[t_k, t_{k+1})}(t), \quad
t \geq t_0.
\ee
An associated trace is given by
\be
\label{tPL}
A(\tau) = \sum_{k=-1}^\infty \log t_{k+1} \one_{[\len_k, \len_{k+1})}(\tau),
\quad \tau \geq 0. 
\ee 
Here $\len_{-1}=0$. 
Conversely, if $A$ is defined by \qref{tPL}, then $\len$ is given by
\qref{lPL}. We fix $q \in (c,\infty)$, assume \qref{lfinite}, and
use \qref{lPL} and \qref{tPL} to compute the integrals in
\qref{lAfinite}. Both equal
\[ e^{-q\len_0} \log \left(\frac{t_1}{t_0}\right) + e^{-q\len_1} \log
\left(\frac{t_2}{t_1}\right)  + e^{-q\len_2} \log
\left(\frac{t_3}{t_2}\right) + \ldots \]

{\em 2.\/} Suppose $\len$ is given, and \qref{lfinite} holds. Fix $q
>0$. We approximate $\len$ by a decreasing sequence of piecewise constant
functions $\len_n \dnto \len$. Then $e^{-q\len_n} \upto e^{-q\len}$
since $q > c \geq 0$, and moreover we have $A_n \upto A$ by
Lemma~\ref{inverseA}.  By the monotone convergence theorem,
\ba
\nn
\lefteqn{ \int_{t_0}^\infty e^{-q \len(t)} \frac{dt}{t} = \lim_{n \to \infty}
\int_{t_0}^\infty e^{-q\len_n(t)} \frac{dt}{t}}\\
\nn
 && = \lim_{n \to \infty}
\int_0^\infty e^{-q\tau} A_n(d\tau) = \int_0^\infty e^{-q\tau}
A(d\tau). 
\ea

{\em 3.\/} To prove (b), assume $A$ is such that $\alphabar(q) <
\infty$ for $q \in (c, \infty)$. Consider a 
sequence of piecewise constant increasing functions $A_n \upto A$. Then
$\len_n \dnto \len$ and we may apply the monotone convergence theorem again.
\end{proof}

Proposition~\ref{propCV} and Lemma~\ref{inverseA}
allow us to reformulate the classical
equivalence between weak convergence of measures and pointwise
convergence of Laplace transforms. The following theorem is a slight
modification of~\cite[XIII.1.2a]{Feller}. 
\begin{thm}
\label{WT}
Suppose $\len_n$ is a sequence of min histories on
$[t_0,\infty)$ with associated traces $A_n$ that satisfy
\be
\label{uniform}
\sup_n \alphabar_n(q) < \infty, \quad q \in (c, \infty),
\ee
for some $c\ge0$.
Then there is a min history $\len$ such that 
$\len_n(t)\to\len(t)$ as $n\to\infty$ at all points of continuity, 
if and only if there is a trace $A$ associated to $\len$ such that 
$\alphabar_n(q)\to \alphabar(q)$ as $n\to\infty$, for all $q\in
(c,\infty)$.  
\end{thm}
\begin{proof}
Suppose $\len_n\to\len$.
Lemma~\ref{inverseA} and definition
\qref{defA2} then imply the distribution functions $A_n\to A$.
By the classical criterion
for weak convergence of measures \cite[XIII.1.2a]{Feller}, under the
hypothesis \qref{uniform} it follows 
$\alphabar_n(q)\to\alphabar(q)$
for all $q>c$.

Conversely, suppose $\alphabar_n(q)\to\alphabar(q)$ for all $q>c$,
where $A$ is a trace (in particular $A(\tau)\to\infty$ as
$\tau\to\infty$).  By the classical criterion, 
the distribution functions $A_n\to A$, and Lemma~\ref{inverseA}
yields $\len_n\to\len$
where $\len$ is given by \qref{defL2}.
\end{proof}

\subsection{Regular variation}
A measurable function $\rv : (0,\infty) \to (0,\infty)$ is
{\em regularly varying \/} at infinity with index $\theta \in \R$
(written $\rv \in \RV_\theta$) if 
\be
\label{sv2}
\lim_{\lambda \to \infty} \frac{\rv (\lambda x)}{\rv(\lambda)}
=x^\theta, \quad 
\mathrm{for\;\; every}\;\;x >0.
\ee
In this case, $\rv(x) = x^\theta \sv(x)$ where $\sv$ is slowly varying
at infinity.

The class of regularly varying functions is remarkably rigid. 
For example, there is no need to assume that the limit in \qref{sv2} 
exists for every $x>0$.  If $\rv$ is a positive, measurable function
on the half-line and  $\lim_{\lambda \to \infty} \rv(\lambda
x)/\rv(\lambda)$ exists, is positive and finite for $x$ in a set of positive
measure, then $\fun$ is regularly varying at infinity with some
index $\theta \in \R$. This fundamental rigidity
lemma (see \cite[VIII.8.1]{Feller} and \cite[1.4.1]{Bingham})
 plays a key role in our analysis.

The class $\RV_\theta$ is of fundamental utility in Tauberian
arguments linking a measure $\nu$ on $[0,\infty)$ and its Laplace
transform $\bar{\nu}(q) = \int_0^\infty e^{-qx}\nu(dx)$, see
\cite[XIII.5.2]{Feller}: 
\begin{thm}
\label{HLK}
If $\sv$ is slowly varying at infinity and $0 \leq \theta < \infty$,
then the following are equivalent:
\be
\label{sv4}
\nu(x) \sim x^\theta \sv (x), \quad x \to \infty,
\ee
and
\be
\label{sv5}
\bar{\nu}(q) \sim q^{-\theta} \sv (1/q)\Gamma(1+\theta), \quad
q \to 0. 
\ee
Moreover, this equivalence remains true when we interchange the roles
of the origin and infinity, namely when $x \to 0$ and $q \to \infty$.
\end{thm}

A refinement of this result will prove useful for 
us---de Haan's exponential Tauberian theorem \cite[Thm.
3.9.3]{Bingham}, see \cite{deHaan}.
\begin{thm}\label{t.deH}
$\exp \nu(x)$ is regularly varying at infinity with index $\theta$
if and only if $\exp \bar\nu(q)$ is regularly varying at zero
with index $\theta$.
If either holds, then 
\[
\nu(1/q)-\bar\nu(q) \to \gamma \theta, \quad\mbox{$q\to0$,}
\]
where $\gamma= 0.577215665\ldots$ is the Euler-Mascheroni constant.
\end{thm}

\subsection{Rate of divergence in the solution formula}\label{s.kappa}
For several reasons, we need to study carefully 
the asymptotic behavior of $\forma(z)$ as $z \to 1$.  
For this, it is convenient to define $\kappa(q)$ via
\begin{equation}\label{kappadef}
- \log \kappa(q)= \forma(1-q)+\log q = 
\int_0^{1-q} \left(
\frac{Q_1}{1-Q(z)} - \frac{1}{1-z}\right) \,dz 
\end{equation}
(note $-\log\kappa(q)\ge 0$) and set
\begin{equation}\label{kappanod}
\kappa_0:= \lim_{q \to 0} \kappa(q) \in [0,1]\,.
\end{equation}
With this notation $\kappa_0$ corresponds to the 
number $\kappa$ introduced in \cite{GM}.
For $Q(z)=z^2$ we have $\kappa_0=\frac12$.

We will show that $\kappa_0>0$ if and only if $\sum_{k=1}^{\infty}
(k \log k) p_k < \infty$.  Since in \cite{GM} the function $Q$ is a
polynomial this finiteness condition is always satisfied. 
However, we can characterize self-similar solutions and their
domains of attraction also  in the case $\kappa_0=0$. 

\begin{lemma}\label{L.kappa}
The function $\kappa(q)$ as defined in \eqref{kappadef} satisfies
$\kappa(q) \to 0$ as $q \to 0$ if and only if
$\sum_{k=1}^{\infty} p_k k \log k = \infty$.
\end{lemma}
\begin{proof}
We compute
\be\label{e.kap1}
R(z):=Q_1-\frac{1-Q(z)}{1-z} = 
\sum_{k=1}^\infty p_k 
\left( k -\frac{1-z^k}{1-z} \right) = 
\sum_{k=1}^\infty p_k  \sum_{j=1}^{k-1} (1-z^j).
\ee
Note $R(1)=0$ and the integrand in \qref{kappadef} is
\be\label{e.kap2}
\frac{R(z)}{1-Q(z)} = 
\frac{R(z)}{1-z}\cdot\frac1{Q_1-R(z)} = \frac{R(z)}{1-z}\cdot\frac1{Q_1+o(1)}
\ee
as $z\to1^-$. Thus it suffices to show $\int_0^1R(z)\,dz/(1-z)<\infty$
if and only if $\sum p_kk\log k<\infty$. We observe
\be
\int_0^1 \sum_{j=1}^{k-1} \frac{1-z^j}{1-z}\,dz = 
\sum_{j=1}^{k-1} \sum_{l=1}^j \frac1l
= \sum_{l=1}^{k-1} \sum_{j=l}^{k-1} \frac1l
= \sum_{l=1}^{k-1} \frac{k-l}l \sim k\log k
\ee
as $k\to\infty$, whence the desired result follows from \qref{e.kap1}.
  \end{proof}

\begin{lemma}\label{L.kappa2}
The function $\kappa(q)$ is slowly varying as $q \to 0$. 
\end{lemma}
\begin{proof}
We need to show that $- \log (\kappa(\lambda q)/{\kappa(q)}) \to
0$ as $q \to 0$ for any $\lambda >0$.
From \qref{e.kap1} and \qref{e.kap2} above, we see
that if we replace the integrand in \qref{kappadef} by $R(z)/(1-z)$, then in
integrating from $1-\lambda q$ to
$1-q$ the error is only $o(1)$ as $q \to 0$.
Hence by \qref{e.kap1},
\[
- \log \frac{\kappa(\lambda q)}{\kappa(q)}  = \sum_{k=1}^{\infty} p_k
k \sum_{l=1}^{k-1} \Big(\frac{ (1-\lambda q)^l}{l} 
- \frac{(1-q)^l}{l}\Big) + o(1)
\]
as $q \to 0$.
We consider without loss of generality $\lambda >1$ and estimate
\[
 \sum_{l=1}^{k-1} \frac{ (1-\lambda q)^l}{l}
 - \frac{(1-q)^l}{l} = \int_{1-\lambda q}^{1-q} \sum_{l=1}^{k-1}
 x^l\,dx = \int_{1-\lambda q}^{1-q} \frac{1-x^{k-1}}{1-x} 
 \,dx \leq \log \lambda \,.
 \]
 On the other hand
 \[
  \sum_{l=1}^{k-1} \frac{ (1-\lambda q)^l}{l}
   - \frac{(1-q)^l}{l} \to 0
   \]
as $q \to 0$ for any $k$  and thus the claim follows from the
dominated convergence theorem.  
\end{proof}
%%%%%%%%%%%%%%%%%%%end lt%%%%%%%%%%%%%%%%%%%%
%%%%%%%
%%%%%%% input ds%%%%%%%%%%%%%%%%%%%%%%%%%%%
\section{Well posedness for measures}
\label{sec.ds}
In this section, we first work with the min-driven time scale
$\len(\tau)=\tau$ used by Gallay and Mielke. We prove that a
continuous probability density with min
$\tau_0>0$ defines a solution to \qref{evol3}. This is a weaker
form of the well-posedness theorem of~\cite{GM}. It is included for its
simplicity, and because it is the basis for weak solutions. We then
switch to the number-driven time 
scale \qref{intime} and use the moment identity \qref{moment2} to show
that \qref{evol3} defines a continuous dynamical system on
the space of probability measures $\PPE$.

\subsection{Classical solutions}
The solution formulas of section~\ref{subsec:GM}, while explicit, are
not immediately suited for the construction of solutions. The main
difficulty is to show that positive initial data yields a positive
solution. We construct solutions by rewriting \qref{evol3} in integral form
\ba
\label{intrho}
\rho_\tau(x) & = & \rhoinit(x) + \int_{\tau_0}^\tau \rho_s(s)
\Q(\rho_s) (x-s) \,  ds, \quad x > \tau, \\
\nn
\rho_\tau(x) & = & 0, \quad x < \tau. 
\ea
We then have 
\begin{thm}
Suppose $\rhoinit$ is a continuous probability density with positive
min $\tau_0$. There exists a unique solution to \qref{intrho} on
$[\tau_0, \infty)$ such that $\rho_{\tau_0}=\rhoinit$ and the
solution has the following properties. 
\begin{enumerate}
\item[(a)] For every $\tau \geq \tau_0$,  $\rho_\tau$ is a continuous
  probability density with min $\tau$. 
\item[(b)] The solution formula \qref{sol5} holds for every $q \geq 0$,
  $\tau \geq \tau_0$. 
\end{enumerate}
\end{thm}
\begin{proof}
\nwc{\pp}{{\rho}}
\nwc{\reals}{\mathbb{R}}
\nwc{\dd}[1]{{\,{\rm d}#1}}
\nwc{\pa}{\partial}
We sketch a proof of existence similar to the direct approach in
\cite{CP} for $Q(z)=z^2$ using a different time scale.
We fix $\tau_0>0$ and let $\rhoinit$ be given.
Note that since the solution is to satisfy $\pp_\tau(x)=0$ for $x<\tau$,
the convolution term on the right-hand side of \qref{intrho}
will depend only upon values of $\pp_\tau(y)$ for $\tau_0<y<x-\tau\le x-\tau_0$.
In particular, this convolution term vanishes for $x<2\tau_0$.

This means we can construct the solution for
$\tau_0<\tau<2\tau_0$ by an inductive procedure as follows: 
For $\tau_0<\tau\le x<2\tau_0$ we have
$\pp_\tau(x)=\rhoinit(x)$ and in particular $\pp_\tau(\tau)=\rhoinit(\tau)$. 
For $\tau_0<\tau\le 2\tau_0$, successively on strips
$x\in[k\tau_0, (k+2)\tau_0)$, for $k=2,4,\ldots$, by simple integration
in time we can now compute $\pp_\tau(x)$ from \qref{intrho},
where the right-hand side is always known from a previous step.
This determines $\pp_\tau(x)$ for $\tau_0\le \tau\le 2\tau_0$ and all $x$.

To determine the solution globally for all $\tau>\tau_0$,
the idea is to replace $2\tau_0$ by $\tau_0$ and repeat.
But in order to justify this we need to verify that 
$\rho_\tau$ remains integrable and conserves total probability.
In particular we need to justify \qref{LT}.
Let us introduce the distribution function
\begin{equation}\label{Fdef}
R_\tau(x)=\int_0^x\pp_\tau(y)\,dy = \int_\tau^x \pp_\tau(y)\,dy.
\end{equation}
This is the probability that a domain has size $\le x$ at time $t$.
Note that for any two distribution functions $R(x)$, $\hat R(x)$ on
$[0,\infty)$ we have
\[
R\star\hat R(x)=\int_0^x R(x-y)\hat R(dy) \le \int_0^x R(x)\hat
R(dy)=R(x)\hat R(x).
\]
Integrating the convolution term in the integrand of \qref{intrho}, we find
\begin{eqnarray*}
&&\int_0^x \Q(\rho_\tau)(y-\tau)\dd y = 
\sum_{k=1}^\infty p_k R_\tau^{\star k}(x-\tau) \le
\sum_{k=1}^\infty p_k R_\tau(x)^k = Q(R_\tau(x)). 
\end{eqnarray*}
Then it follows
$ \pa_\tau R_\tau(x) \le \pp_\tau(\tau)(Q(R_\tau(x))-1)$,
and since $\pp_\tau(\tau)\ge0$ and $R_\tau(x)\le 1$ initially,
$R_\tau(x)$ is decreasing in $\tau$ for fixed $x$.
It follows $R_\tau(\infty)\le1$, and so the Laplace transform
\[
\bar R_\tau(q)=\int_0^\infty \rme^{-qx} R_\tau(\dd{x}) = 
\int_\tau^\infty \rme^{-qx}\pp_\tau(x)\dd{x}
\]
is well defined and $\le\rme^{-q\tau}$.
Since $\pa_\tau R_\tau(x)$ is continuous in $\tau$ for all $x$, 
$\bar R_\tau(q)$ is $C^1$ in $\tau$ for all $q>0$.
This justifies \qref{LT} and the the computations leading up to 
\qref{sol4} and the solution formula \qref{sol5}.
From \qref{sol4} we deduce that since $\bar R_{\tau_0}(0)=1$,
$\bar A(0^+)=\infty=\bar A_\tau(0^+)$ and then \qref{sol5} yields
$\bar R_\tau(0^+)=1=R_\tau(\infty)$, proving that $\pp_\tau$ is a probability density
for all $\tau$.
\end{proof}

The next lemma provides a uniform estimate for smooth approximations.
\begin{lemma}
Suppose $\rhoinit$ is a continuous probability density with positive
min $\tau_0$ and $\rho_\tau$ is the solution to \qref{intrho}
with initial data $\rhoinit$. Then (with $\alpha_0=0$ in \qref{alphadef}),
\be
\label{Aest}
A(\tau) \leq \frac{Q_1}
{\log 2} 
{\log \left(\frac{2\tau}{\tau_0}\right)}
.
\ee
Consequently, 
\be
\label{Lest}
\frac{\len(t)}{\len(t_0)} \geq   
\frac12 \left( \frac{t}{t_0} \right)^{(\log 2)/Q_1}, \qquad  t\ge t_0.
\ee
\end{lemma}
\begin{proof}
If $m$ is an integer such that  $\tau \in [2^{m-1}\tau_0, 2^m \tau_0)$
we divide the domain of integration $[\tau_0,\tau)$ into $m$ pieces
and use the fact 
that $\rho_s(s)= \rho_{r}(s)$ for  $\tau_0\vee\frac12 s\le r\le s$ to obtain 
\ba 
\nn
Q_1\inv A(\tau) &=& 
\int_{\tau_0}^{2\tau_0}\rho_{\tau_0}(s) ds  +
\int_{2\tau_0}^{2^2\tau_0} \rho_{2\tau_0}(s) \, ds + \ldots
\int_{2^{m-1}\tau_0}^\tau \rho_{2^{m-1}\tau_0}(s) \, ds\\
\nn
& 
\leq & m \ \leq \frac{\log (\tau/\tau_0)}{\log 2} +1. 
\ea
\end{proof}

\subsection{Weak solutions}
We now switch to the number-driven time scale $N(t)=1/t$. In order to
define weak solutions, we fix a test function $a$, substitute
\qref{rholdot} in the moment identity \qref{moment2} and integrate in
time between $t_0$ and $t$ to obtain
\ba
\label{newmoment}
\lefteqn{ \int_{\Rplus} a(x)F_t(dx) - \int_{\R_+} a(x) F_{t_0}(dx) =} \\
&&
\nn
\int_{t_0}^t \sum_{k=1}^\infty p_k \int_{\Rplus^k} \left[ a\left(\len(s) +
     \sum_{i=1}^k y_i \right) -a(\len(s))\right] \prod_{i=1}^k
 F_s(dy_i)\, \frac{ds}{Q_1s}.
\ea
We will consider continuous test functions with
$\lim_{x \to \infty} a(x)=0$. Let $C_0(\Rplus)$ denote the space of
such functions with the topology of uniform
convergence. Let $\PPE$ denote the space of probability measures
on $\Rplus$ equipped with the weak topology. Assume $t_0 > 0$ is fixed.

\begin{defn} Let $J\subset (0,\infty)$ be an interval.
We say that a map $F: J \to \PPE$ is a weak solution for
min-driven clustering on $J$ if
\begin{enumerate}
\item The map $t \mapsto \int_{\Rplus} a(x) F_t(dx)$ is measurable 
for every $a \in C_0(\Rplus)$. 
\item The min of $F_t$, denoted $\len(t)$, is positive and increasing.
\item The moment identity \qref{newmoment} holds for each 
$a \in C_0(\Rplus)$ and $t, t_0 \in J $.
\end{enumerate}
\end{defn}

\begin{thm}
\label{thm.wp}
\begin{enumerate}
\item[(a)] 
Suppose $\hat{F} \in \PPE$ has positive min, and $t_0>0$.  Then 
there is a weak solution $F$ for min-driven clustering on $[t_0,\infty)$
with $F_{t_0}=\hat{F}$. Moreover, the min $\len(t)$ of $F_t$ 
satisfies \qref{Lest}.
\item[(b)] The solution in (a) is unique on $[t_0,t_1]$ for any $t_1>t_0$.
\item[(c)] Let $\hat{F}\nsup$ be a sequence in $\PPE$ with positive
  min and $F\nsup$ the weak solutions with
  $F\nsup_{t_0}=\hat{F}\nsup$. Assume $\lim_{n \to \infty}
  \hat{F}\nsup=\hat{F}$ and the limit has positive min. Then $F\nsup_t
  \to F_t$ for every $t > t_0$.
\end{enumerate}
\end{thm}
\begin{proof}
{\em 1.\/} It follows from Weierstrass' approximation theorem that finite 
linear combinations $\sum_{k=1}^N c_k e^{-kx}$ are dense in
$C_0(\Rplus)$. Therefore, in order to verify the moment identity, it
is sufficient to consider the test functions $e^{-qx}$, $q>0$.
Thus, to prove existence of a weak solution on $[t_0,\infty)$, 
it suffices to construct $F$ weakly continuous
such that the Laplace transform
satisfies the solution formula \qref{newsol5}. 

{\em 2.\/} Let $\tau_0$ denote the min of $\hat{F}$.  We approximate
$\hat{F}$ by a sequence of continuous probability densities
$\hat\rho\nsup$ with min $\tau_0\nsup$ with $\lim_{n \to
  \infty} \tau_0\nsup= \tau_0$. We further assume that $\hat\rho\nsup$
is strictly positive on $[\tau_0\nsup, \infty)$. It is immediate from
\qref{intrho} that the solutions $\rho_\tau\nsup$ are strictly
positive on $[\tau,\infty)$. The trace for these solutions, $A\nsup$
is obtained from \qref{alphadef} with $\tau_0$ replaced by 
$\tau_0\nsup$, $\rho$ by $\rho\nsup$ and the choice $\alpha_0=\log
t_0$. $A\nsup$ is continuous and strictly increasing, thus so are the
min histories $\len\nsup$.  We change variables from the solution
formula \qref{sol5} to \qref{newsol5} to obtain 
\be
\label{approxsol}
\rhobar\nsup_\tau(q) = \Rbar_t\nsup(q) =
 \forma^{-1} \left( \int_t^\infty e^{-q \len\nsup(s)}\frac{ds}{
     s}\right), \quad t \geq t_0.
\ee

{\em 3.\/}  As $n \to \infty$, we have $\hat\rhobar\nsup(q) \to
\Rbarhat(q)$ for every $q>0$ and $0 < \Rbarhat(q) < 1$ for $q>0$. 
The measures $A\nsup$ are supported on $[\tau_0\nsup,\infty)$ and
satisfy  $\alphabar\nsup(q) =\forma(\Rbar_{t_0}\nsup(q))$. Therefore, 
\[ \lim_{n \to \infty} \bar{A}\nsup(q) =\forma(\hat{\Rbar}(q)), \quad
q>0. \]
It follows that the traces $A\nsup$ converge  weakly to a
trace $A$ supported on $[\tau_0,\infty)$ that satisfies
$\alphabar(q) = \forma(\hat{\Rbar}(q))$. Therefore, by
Theorem~\ref{WT} the min histories
$\len\nsup$ converge to $\len$, the min history
associated to $A$.  Since $\len\nsup$ satisfy the uniform estimate
\qref{Lest},  we may use  the dominated convergence theorem
to assert 
\[ 
\lim_{n \to \infty} \Rbar_t\nsup(q) = 
\forma\inv\left( 
\int_t^\infty e^{-q \len(s)}\frac{ds}{s} \right)
, \quad t \geq t_0. 
\]
This shows that for every $t \geq t_0$ the measures $F_t\nsup$
converge weakly to a measure $F_t$ that satisfies \qref{newsol5}.
Then $F_t$ is a probability measure since
$\Rbar_t(0+)=\forma\inv(\infty)=1$.
It similarly follows from \qref{newsol5} that 
$F_t \to F_{t_1}$ as $t\to t_1$ for every $t_1 \in [t_0,\infty)$. 
This completes the proof of part (a), except that it remains to 
show that $\len(t)$ is in fact the min of $F_t$. 

{\em 4.\/} For (b) it suffices to prove uniqueness on 
$[t_0,t_1]$ for {\em some} $t_1>t_0$. 
The key is to prove uniqueness of the min $\len(t)$ on such a time
interval, since then uniqueness of $F_t$ is easy to establish via 
the Laplace transform.
Note that any weak solution is weakly continuous in time, since 
the right-hand side of \qref{newmoment} is Lipschitz continuous for
any test function $a\in C^0(\Rplus)$. Then since $F_{t_0}(x)>0$
for all $x>\tau_0$, the min $l(t)$ is right continuous at $t_0$,
so there exists $t_1>t_0$ with $\tau_0\le\tau_1=l(t_1)<2\tau_0$.
Now the idea is that for clusters of size less than $2\tau_0$, there
is no gain, only loss.  We claim that 
\begin{equation}\label{e.smallx}
F_t(x) = 
\left( F_{t_0}(x)- \frac1{Q_1}\log\left(\frac{t}{t_0}\right)\right)_+, 
\quad t\in[t_0,t_1], \quad x\in[\tau_0,2\tau_0).
\end{equation}
This follows by first considering points $x$ of continuity of both
$F_t$ and $F_{t_0}$ such that $\len(t)<x<2\tau_0$, so $F_t(x)>0$ --- 
approximate $\one_{[0,x)}$ by continuous test functions $a$ 
supported in $[0,2\tau_0)$.
By right continuity and the definition of min, 
\qref{e.smallx} follows for all $x\in[\len(t),2\tau_0)$, and
$F_t(x)=0$ for $x<\len(t)$. Now by weak continuity in time, 
we infer that with
\[
\hat t(x) = t_0\exp(Q_1\hat F(x))  ,
\]
for each point of continuity of $\hat F$ in 
$(\tau_0,\tau_1]$ we have $\len(t)<x$ for
$t < \hat t(x)$, and $x<\len(t)$ for $\hat t(x)<t$.
Now clearly $\len(t)$ is uniquely determined by
$\hat F$, since it is locally the inverse of $\hat t$.
It follows from taking the Laplace transform that the 
min history $\len(t)$ constructed in (a), that satisfies
\qref{newsol5}, agrees with the min of $F_t$.

{\em 5.\/} Part (c) is proven by
an argument very similar to Step 2 above. 
\end{proof}

{\em Scaling.\/}
We note for use below the following scaling property that follows
easily from the moment identity \qref{newmoment}:
Let $a,b>0$.
If $F$ is a weak solution for min-driven clustering on an
interval $J$, then $\hat F$ is a weak solution on $J/a$, where
\begin{equation} \label{e.scale}
\hat F_t(x) = F_{at}(bx), \qquad t\in J/a,\ x\ge 0.
\end{equation}
%%%%%%%%%%%%%%%%%%%end ds   %%%%%%%%%%%%%%%%%%%%
%%%%%%%
%%%%%%% input ss%%%%%%%%%%%%%%%%%%%%%%%%%%%
\section{Self-similar solutions} \label{sec.ss}

Let us recall that Gallay and Mielke have classified the
self-similar solutions as in \qref{ftheta1} in the case when 
$p_k=0$ for all large $k$. Here we will recover and extend the
basic existence results and classify all domains of attraction
by simple means based on the Laplace transform and Tauberian arguments. 

{\em Existence.}
In terms of the distribution function $F_t$, self-similar solutions
are of the form $F_t(x)= F_*(x/l(t))$ for some distribution function
$F_*$ with positive min ($=1$) and some min history $l(t)$. 
Without loss of generality we can assume $t_0=1$ and $l(t_0)=1$. 
From \eqref{newsol5} we obtain for the Laplace transform of $F_*$ that
\begin{equation}\label{ss1}
\bar F_*(q) = \bar F_t\left( \frac{q}{l(t)}\right) = \varphi^{-1} 
\left( \int_t^{\infty} e^{-q \frac{l(s)}{l(t)}} 
\frac{ds}{s}\right) = 
\varphi^{-1} \left( \int_1^{\infty} e^{-q \frac{l(ts)}{l(t)}}
\frac{ds}{s}\right).
\end{equation}
Hence we conclude that the min history of self-similar solutions 
must satisfy ${l(ts)}/{l(t)}= g(s)$ for some function $g(s)$. 
Since $l$ is also positive, increasing and non-constant, we conclude 
(see \cite[VIII.8.1]{Feller}, or subsection~\ref{ss.nec} below) that 
necessarily, for some $\theta>0$, $l(s)=s^{1/\theta}$ and 
that $F_*$ must be a distribution $F\thetasup$ that satisfies
\be\label{sss1b}
\Rbar\thetasup(q) = \forma\inv\left( \theta \,\mathrm{Ei}(q)\right),
\ee
where $\mathrm{Ei}(q) = \int_q^\infty e^{-s}ds/s$ denotes the exponential
integral. 
Provided such a distribution $F\thetasup$ does exist, then 
$\Rbar_t(q)=\Rbar\thetasup(t^{1/\theta}q)$ satisfies \qref{newsol5}
and is continuous in time, hence determines a self-similar
solution by step 1 of the proof of Theorem~\ref{thm.wp}.
The corresponding trace measures are given by
\be
\label{sss1}
A\thetasup(\tau) = \theta\, \log \tau, \quad \tau \in
(0,\infty).
\ee
As a byproduct of the characterization of scaling limits in section
\ref{sec.rv}, we will prove the existence of distributions
$F\thetasup$ that satisfy \qref{sss1b} for $0<\theta\le1$, and that
$\theta\le1$ is necessary.

{\em Densities.}
Next we show that for the 
self-similar solutions $F\thetasup(x/t^{1/\theta})$  $(0<\theta\le1)$,
the probability distributions $F\thetasup$ have piecewise smooth
densities $\den\thetasup$ that satisfy an integrodifferential equation,
\begin{equation}\label{etaeq}
- \partial_y \big ( y \den^{(\theta)}(y) \big ) = 
\frac{\theta}{Q_1} \,\Q(\den^{(\theta)})(y-1)\,,\quad y>1, 
\qquad \den^{(\theta)}(1)=\frac{\theta}{Q_1}\,,
\end{equation}
and we study the decay of $\den\thetasup(y)$ as $y\to\infty$.
When $p_k=0$ for all large $k$ (so $Q(z)$ is a polynomial)
this has been done by Gallay and Mielke \cite{GM}. 
We will recover most of their results by simpler means
(the exponential rate of decay for $\theta=1$ is an exception)
and extend them to the case when $p_k$ is nonzero for infinitely many $k$.
It turns out, however, that when $\sum p_kk\log k=\infty$, {\em none} 
of the profiles have finite first moment, including the case $\theta=1$.

Let $\theta\in(0,1]$. We claim that the probability distribution $F=F\thetasup$ 
satisfies the following weak-form profile equation:
\be\label{e.wkss}
\int_{\R_+} x a'(x) F(dx) = 
 \frac{\theta}{Q_1}  \sum_{k=1}^\infty 
p_k \int_{\R_+^k} 
\left( a\left(1+\sum_{j=1}^k y_j\right) - a(1) \right) \prod_{j=1}^k F(dy_j) 
\ee 
for all $C^1$ functions $a\in C_0(\R_+)$.
This follows from the fact that we know $F(x/t^{1/\theta})$ is a weak
solution for the min-driven clustering equation \qref{newmoment} with
$\len(t)=t^{1/\theta}$. Changing variables in \qref{newmoment} via
$x=\len(t)\hat x$ and similarly for $y_j$, we differentiate at $t=1$
to obtain \qref{e.wkss}.

Now, the min of $F$ is 1 (this will be shown in section 6.2), 
so $F$ has density $\den(x)=0$ on $(0,1)$.  Taking $a$ to be
supported in $(1,2)$ we find that the right-hand side of \qref{e.wkss}
vanishes. Hence restricted to $(1,2)$, the measure 
$xF(dx)=\beta\,dx$ for some constant $\beta$, so $F$ has density 
$\den(x)=\beta/x$ on $(1,2)$.
Taking $a(x)=0$ for $x\le1$, $a(x)=1$ for $x\ge2$ 
(approximated by limits---note $\int_n^{n+1}xF(dx)\to0$ along a
subsequence), 
we find $\beta = \theta/Q_1$. Taking $a$ supported in $(0,2)$ ,
we find the right-hand side is $-\beta a(1)$ and we
can conclude that $1$ is not an atom of $F$. 

In this way, 
proceeding inductively on intervals $(0,n)$ we deduce that 
$F$ has density $\den(x)$ satisfying \qref{etaeq}.
When $a$ has support in $(1,n+1)$ the right-hand side depends on
the restriction of $F$ to $(0,n)$ where it has density $\den(x)$,
and thus $xF(dx)$ has density $x\den(x)$ determined by 
\qref{etaeq} on $(1,n+1)$.

{\em Decay.}
Next we wish 
to characterize the decay behavior of the densities $\den\thetasup$.
For this we use the properties of the function $\kappa(q)$ 
which was introduced in \eqref{kappadef}.
We denote by $\kappa^\#$ the {\em de Bruijn conjugate} of $\kappa$. 
See \cite[Sec.\ 1.5.7]{Bingham}.
This is a slowly varying function satisfying
\be \label{kapsharp}
\kappa(q)\kappa^\#(q\kappa(q))\sim 1 \quad\mbox{ as $q\to0$.}
\ee
If $\kappa_0>0$, then $\kappa^\#(0^+)=\kappa_0\inv$, 
and if $\kappa_0=0$, then $\kappa^\#(0^+)=\infty$.

\begin{thm} \label{t.ss}
For every $\theta \in (0,1]$,
the density $\den^{(\theta)}(x)$ of the self-similar profile
$F\thetasup$ has the following properties:
\begin{enumerate}
\item[(i)]
If $\theta 
\in (0,1)$, then as $x\to\infty$,
\begin{equation}
\label{etadecay}
\den^{(\theta)}(x) \sim x^{-(1+\theta)} e^{\theta \gamma}\kappa^\#(x^{-\theta}) 
\frac{\theta(1-\theta)}{\Gamma (2-\theta)}.
\end{equation}
Here $\gamma$ 
is the Euler-Mascheroni constant and $\Gamma$ is the $\Gamma$-function.
\item[(ii)]
If $\theta =1$ then as $x \to \infty$,
\begin{equation}
\label{Fdecay}
\int_0^x y \den^{(1)}(y)\,dy \sim  { e^{ \gamma}}{\kappa^\#(x^{-1})} \,
\end{equation}
\end{enumerate}
\end{thm}
{\em Remarks:}
The asymptotics \eqref{Fdecay} imply the result of \cite{GM} for total
mass in the case that $\kappa_0>0$.  We will not pursue here the
delicate question of the precise (exponential) 
decay rate of the density $\den^{(1)}$ in this case, which was studied in 
\cite{CP} for $Q(z)=z^2$ and in the polynomial case in \cite{GM}.
If $\kappa_0=0$, however, we see that $F^{(1)}$ has infinite mass. 

\begin{proof}
We rewrite equation \eqref{sss1b} for $ \Rbar\thetasup$ 
using \qref{kappadef} and standard asymptotics for the 
exponential integral \cite{Olver}, as follows. With $w=1-\Rbar\thetasup(q)$, 
\[
-\forma(\Rbar\thetasup(q))= \log(w\kappa(w)) = -\theta\mathrm{Ei}(q)
=\theta(\log q+ \gamma+o(1)), \quad q\to 0.
\]
Then $w\kappa(w)\sim q^\theta e^{\theta\gamma}$,
whence asymptotic inversion \cite[Sec.\ 1.5.7]{Bingham} yields 
\be \label{asyinv}
w=1-\Rbar\thetasup(q) \sim q^\theta e^{\theta\gamma} \kappa^\#(q^\theta) .
\ee
Now differentiating (i.e., using Lemma 3.3 of \cite{MP1}) we find
\[
\int_0^{\infty} e^{-qx}x F^{(\theta)}(dx)  =
-\partial_q  \Rbar\thetasup (q) \sim \theta q^{\theta-1} 
{e^{\theta \gamma}}{\kappa^\#(q^{\theta})}, \qquad q\to0.
\]
Now the Tauberian theorem \ref{HLK} implies 
\begin{equation}
\label{Fdecay1}
\int_0^x y F^{(\theta)}(dy) \sim 
\frac\theta{\Gamma (2-\theta)}
x^{1-\theta} 
{e^{\theta \gamma}\kappa^\#(x^{-\theta})} 
\,, \qquad x\to\infty,
\end{equation}
which is just \eqref{Fdecay} in the case $\theta=1$.

We can now also derive the decay behavior of $\den^{(\theta)}$ in the
case $\theta<1$. In fact, it follows from \eqref{etaeq}
that $y  \den^{(\theta)}(y)$ is decreasing. Hence by a lemma in
\cite[XIII.5]{Feller},  \eqref{Fdecay1} implies  \eqref{etadecay}.
\end{proof}
%%%%%%%%%%%%%%%%%%%end ss  %%%%%%%%%%%%%%%%%%%%
%%%%%%%
%%%%%%% input rv  %%%%%%%%%%%%%%%%%%%%%%%%%%%
\section{Domains of attraction of self-similar solutions}
\label{sec.rv}

We proceed to prove Theorem~\ref{thm.domains}. 
The proof has two parts. The first
is to show that regular variation of $\len(t)$ as $t \to \infty$ is
necessary and sufficient for convergence to a scaling limit. 
The second is the equivalence between 
regular variation of $\len(t)$ as $t \to \infty$ and  
regular variation of $\int_0^x y\newrho_{t_0}(dy)$ as $x \to \infty$. 
The second part is based on the Tauberian theorems~\ref{HLK}
and~\ref{t.deH}. 
The most subtle aspect (despite the simple proof) is to deduce regular
variation of $\len$ from the existence of a scaling limit. 
This is the assertion of rigidity, and we treat it first.

\subsection{Regular variation of the min history is necessary}
\label{ss.nec}
{\em 1.\/} 
Assume there is a rescaling $\lambda(t) \to \infty$ and a
probability distribution function $\newrho_*$ with positive min
(called $\tau_*$) such that $F_t(\lambda(t) \cdot) \to F_*$. 
This is equivalent to convergence of the Laplace transforms,
\be
\label{sss10}
 \lim_{t \to \infty} \Rbar_t \left( \frac{q}{\lambda(t)} \right) =
\Rbar_*(q), \quad q \geq 0. 
\ee
After a trivial scaling of time and cluster size, 
we may assume $t_0=1$ and $\tau_*=1$.
Since $F_*$ is a probability measure with positive min,
there is a unique trace measure $A_*$, with $A_*(\tau)=0$ on
$(0,1)$ and $A_*(\tau)\to\infty$ as $\tau\to\infty$,
and a min history $\len_*$ on $[1,\infty)$, with $\len_*(1)=1$,
such that
\be
\label{sss2} \alphabar_*(q)= \int_{1}^\infty e^{-q\len_*(s)}\frac{ds}{s}=
\forma(\Rbar_*(q)).
\ee
Now, $s\mapsto F_{ts}(\lambda(t)x)$ is a rescaled solution, by \qref{e.scale}.
The rescaled min histories given by $\len^{(t)}(s)=\len(ts)/\lambda(t)$,
$t,s\ge 1$ have associated trace measures $A^{(t)}$ satisfying 
$\alphabar^{(t)}(q) \to \alphabar_*(q)$ as $t\to\infty$, for all $q>0$.
We use the solution
formula \qref{newsol5}, \qref{sss10} and \qref{sss2} to obtain
\be
\label{sss3}
\lim_{t \to \infty} \int_{1}^\infty e^{-q\len(ts)/\lambda(t)}
\frac{ds}{s} = \int_{1}^\infty e^{-q\len_*(s)}\frac{ds}{s}, \quad q>0.
\ee
Theorem~\ref{WT} now implies  
that at every point of continuity of $\len_*$,  
\be
\label{sss4}
\lim_{t \to \infty} \len(ts)/\lambda(t) = \len_*(s).
\ee

{\em 2. \/}  Let $s_0$ be a point of continuity of $\len_*$,
and let $U(t)=\len(ts_0)$, $t\ge1$. Let $B$ be the set of $x\ge1$
such that 
\be \label{limrat}
\psi(x) = \lim_{t\to\infty} \frac{U(tx)}{U(t)}
\ee
exists.  We deduce $\psi(x)$ is a power of $x$ by 
following the simple argument in \cite[VIII.8.1]{Feller}:
By \qref{sss4}, 
we have $x\in B$ if $xs_0$ is a point of continuity of
$\len_*$. If $x_1$, $x_2\in B$ then $x_1x_2\in B$ and 
\be \label{cauchy}
\psi(x_1x_2)=\psi(x_1)\psi(x_2). 
\ee
Since $U$ is increasing, so is $\psi$, and since $B$ is dense
in $[1,\infty)$ we can extend $\psi$ by right continuity so
\qref{cauchy} holds for all $x_1, x_2\ge 1$. Then we can set
$\psi(x)=1/\psi(1/x)$ for $x\in(0,1)$ and have \qref{cauchy}
for all $x>0$.  Since $\psi(x)$ is positive, locally bounded and not constant,
it follows $\psi$ is a pure power law, and 
we can write $\psi(x)=x^{1/\theta}$ for some $\theta>0$.

{\em 3.\/}  Since $\psi$ is continuous and increasing, it is
easy to see \qref{limrat} holds for all $x\in [1,\infty)=B$.
Then we infer $\tilde L(t)=U(t)t^{-1/\theta}$ is slowly varying, and 
$U$, hence $\len$, is regularly varying at $\infty$ 
with index $1/\theta$.
Further, since $\len_*(xs_0)/\len_*(s_0)=\psi(x)=x^{1/\theta}$
whenever $xs_0$ is a point of continuity of $\len_*$, 
and $s_0$ is an arbitrary point of continuity, 
it follows $\len_*(t)= t^{1/\theta}$ for all $t\ge1$.

\subsection{Regular variation of the min history is sufficient}
\label{subsec:suff}
Let us assume that the min history $\len$ is regularly varying,
as in \qref{conv3}.  Convergence to
self-similar form is then quick: Since 
\[ \frac{\len(ts)}{\len(t)} = s^{1/\theta}
\frac{\tilde{L}(ts)}{\tilde{L}(t)}, \]
we have  
\[ \lim_{t \to \infty}  \frac{\len(ts)}{\len(t)} = s^{1/\theta}.\]
Moreover, for any $\eps >0$, there exists $t_\eps>0$ such that 
\[ s^{1/\theta-\eps} \leq \frac{\len(ts)}{\len(t)}, \quad t \geq t_\eps. \]
Therefore, we may use the solution formula and the dominated
convergence theorem to see that as $t \to \infty$
\be 
\forma\left(\Rbar_t\left(\frac{q}{\len(t)}\right)\right) = \int_{1}^\infty
e^{-q\len(ts)/\len(t)} \frac{ds}{s} \to \int_{1}^\infty
e^{-qs^{1/\theta}} \frac{ds}{s}. 
\ee
This implies $\lim_{t \to \infty} F_t(\len(t) x) =
F\thetasup(x)$ for every $x >0$, where $F\thetasup$ is a probability
distribution with positive min ($=1$) and
Laplace transform given by \qref{sss1b}.

\subsection{Tauberian arguments}

We now prove that if 
$0<\theta\le1$ and $\int_0^x y \newrho_{t_0}(dy)$ is
regularly varying with index $1-\theta$ as $x \to  \infty$, 
then $\len(t)$ is regularly varying with index $1/\theta$
as $t\to \infty$, and conversely. 

{\em 1.\/} 
It is convenient first to assume \qref{conv2} in the form
\[
\int_0^x y F_{t_0} (dy) \sim 
\frac{\theta}{\Gamma(2-\theta)}
x^{1-\theta} L_1(x)
, \quad x \to \infty,
\]
where $L_1$ is slowly varying.
Then Theorem \ref{HLK} implies 
\[
 -\partial_q \Rbar_{t_0}(q) \sim \theta q^{\theta-1} L_1\left(q\inv\right)
, \quad q \to 0.
\]
By integration (not difficult to justify as in Lemma 3.3 of \cite{MP1}) we find
\[
w:= 1- \Rbar_{t_0}(q) \sim q^\theta L_1\left(q\inv\right),
\quad  q \to 0.
\]
Using the solution formula \eqref{newsol5} and
\eqref{kappadef}, we write
\[
\log \big (w \kappa(w) \big )
= - \varphi \big ( \Rbar_{t_0}(q) \big ) = - \int_t^{\infty} e^{-q l(s)} 
\frac{ds}{s} = - \bar A(q)\,.
\]
Hence, $\exp(-\bar A(q))$ is regularly varying with exponent $\theta$,
and Theorem~\ref{t.deH} (de Haan's exponential Tauberian theorem)
implies that
\be \label{linvconv}
e^{A(1/q)}e^{-\bar A(q)} =
\invlen( q\inv ) w \kappa(w) \to e^{\gamma \theta}, \quad q \to 0.
\ee
Thus,
$\invlen(\tau) \sim e^{\gamma\theta} \tau^\theta 
\hat L_1(\tau^\theta)$ as $\tau\to\infty$, with 
\[
\hat L_1(s) = \frac{1}{L_1(s^{1/\theta})
          \kappa(s\inv L_1(s^{1/\theta})) } .
\]
It is easy to show $\hat L_1$ is slowly varying, 
using the uniform convergence theorem for slowly varying 
functions~\cite[Theorem 1.2.1]{Bingham}. 
Asymptotically solving $t=\invlen(\tau)$ by inverting $s\mapsto s\hat L_1(s)$ 
using the de Bruijn conjugate \cite[Theorem 1.5.13]{Bingham} finally yields
\begin{equation}\label{e.len}
l(t) \sim 
e^{-\gamma}t^{1/\theta} 
\hat L_1^\#(t)^{1/\theta}  , \quad 
t \to \infty\,, 
\end{equation}
which gives \eqref{conv3}.
If $\kappa_0 := \lim_{q \to 0} \kappa(q)>0$, we also can write
\begin{equation}
l(t) \sim e^{-\gamma} (\kappa_0t)^{1/\theta} (L_1^{-1/\theta})^\#(t) 
,\quad t \to \infty.
\end{equation}

{\em 2.\/}
We now prove the converse.
Assuming that \qref{conv3} holds, then
\begin{equation}
l(t) \sim e^{-\gamma}t^{1/\theta} L_2(t)^{1/\theta},  \quad  t \to \infty,
\end{equation}
for some slowly varying function $L_2$.
Then, by inversion,
$\invlen(\tau) \sim e^{\gamma\theta}\tau^{\theta} L_2^\# (\tau^{\theta})$ 
as $\tau \to \infty$.
Since $\invlen(\tau) = \exp{A(\tau)}$ and $\bar A(q) = 
-\log (w \kappa(w))$ with $w=1-\Rbar_{t_0}(q)$, 
we infer from Theorem~\ref{t.deH} that \eqref{linvconv} is true.
This implies
\[
w \kappa(w) \sim 
q^{\theta} L_2^\#(q^{-\theta})\inv , \quad q \to 0,
\]
and asymptotic inversion yields
$w\sim  q^\theta \hat L_2(q\inv)$
as $q \to 0$, with  
\newcommand{\tlx}{{L_2^\#(x^{\theta})}}
\[
\hat L_2(x) =
{\tlx\inv}\kappash\left(x^{-\theta}{\tlx\inv}\right),
\]
and $\hat L_2$ is slowly varying.
Differentiating using Lemma 3.3 of \cite{MP1}, we find
\[
\D_qw=-\partial_q \Rbar_{t_0}(q)  \sim  \theta  q^{\theta-1} 
\hat L_2\left(q\inv\right), \quad q\to0.
\]
Since 
$-\partial_q \Rbar_{t_0}(q) = \int_0^{\infty} e^{-qx} x F_{t_0}(dx) $ 
the Tauberian Theorem \ref{HLK} implies that
\begin{equation}
\int_0^x y F_{t_0}(dy) \sim 
\frac{\theta }{\Gamma(2-\theta)} 
x^{1-\theta} \hat L_2(x) ,
\quad x \to \infty\,,
\end{equation}
which is just \eqref{conv2}.
If $\kappa_0:=\lim_{q \to 0} \kappa(q)>0$, then 
\begin{equation}
\int_0^x y F_{t_0}(dy) \sim 
\frac{\theta }{\Gamma(2-\theta)} 
\frac{x^{1-\theta} }{ \tlx\kappa_0}
, \quad x \to \infty\,.
\end{equation}
This completes the proof.
%%%%%%%%%%%%%%%%%%%end rv  %%%%%%%%%%%%%%%%%%%%
%%%%%%%
%%%%%%% input lk  %%%%%%%%%%%%%%%%%%%%%%%%%%%
\section{Eternal solutions}
\label{sec:lk}

\begin{defn}
A weak solution $\newrho$ for min-driven clustering
is an {\em eternal solution\/} if
it is defined on the maximal interval of existence $(0,\infty)$. 
\end{defn}

Our understanding of eternal solutions is closely connected to
the question of how they emerge from clusters of infinitesimal size.
From the solution formula \qref{newsol5} we see that as $t$ approaches
zero,
\[
\varphi(\Rbar_t(q)) = \int_t^\infty e^{-q\len(s)}\frac{ds}s \to
\infty,
\]
hence $\Rbar_t(q)\to1$, and this means the relative size distribution
$F_t$ always converges to a 
Dirac delta at zero size.  A different scaling is needed to
distinguish solutions through limits as $t\dnto0$.

\subsection{The class of \smeass}
What we will show is that the class of eternal solutions 
is in one-to-one correspondence with a suitable space of measures that
can be loosely thought of as `rescaled initial data' at $t=0$. 
This correspondence parallels the classical probabilistic
characterization of infinitely divisible laws via a \LK\/ formula. 
In probability theory, infinitely divisible distributions
are parametrized by the \LK\/ representation theorem, which expresses 
the log of the characteristic function (Fourier transform) in terms of 
a measure that satisfies certain finiteness conditions.
In particular~\cite[XIII.7]{Feller}, a function $\omega(q)$
is the Laplace transform $\intR \rme^{-qx} F(dx)$
of an infinitely divisible probability measure 
$F$ supported on $\ebar$ if and only if
$\omega(q)=\exp(-\lap(q))$ where
the Laplace exponent $\lap$ admits the representation
\be
\label{R.LK2a}
\lap (q) = \int_{[0,\infty)} \frac{1-\rme^{-qx}}{x} \sm(dx)
\ee
for some measure $\sm$ on $[0,\infty)$ that satisfies
\begin{equation}\label{R.gendef2}
\int_{[0,\infty)} (1 \wedge y^{-1}) \sm(dy) <\infty.
\end{equation}
(Here $a \wedge b=\min(a,b)$.)
As in~\cite{MP3} we call such measures {\em \smeass} (short
for ``generating measures,'' a term motivated by their connection
with generators of convolution semigroups in probability
\cite[XIII.9(a)]{Feller}).  Some basic analytic facts
about Laplace exponents and {\smeass\/} are collected
in~\cite[Sec. 3]{MP3}. 

\begin{defn}
\label{R.candef}
A measure $\sm$ on $[0,\infty)$ is a {\em \smeas}\ if 
\qref{R.gendef2} holds.
In addition, we say that a \smeas\ $\sm$ is {\em divergent\/} if
\be
\label{R.candef5}
\sm(0)>0 \quad\mbox{or}\quad 
\int_{[0,\infty)} y^{-1} \sm(dy) = \infty.
\ee
\end{defn}
\noindent

Here recall we use the notation $\sm(x)=\int_{[0,x]} \sm(dy)$.
The space of \smeass\ has a natural weak topology which
is fundamental in our study of scaling dynamics.
\begin{defn}
\la{R.proper}
A sequence of \smeass\ $\sm\nsup$ {\em converges} to a 
\smeas\ $\sm$ as $n\to\infty$, if at every point $x\in(0,\infty)$
of continuity of $\sm$ we have 
\begin{equation}\label{R.candef1}
\sm\nsup(x)\to\sm(x)
\quad\mbox{and }\quad 
\int_{[x,\infty)} y\inv \sm\nsup(dy)\to \int_{[x,\infty)} y\inv \sm(dy) .
\end{equation}
\end{defn}

\subsection{A \LK\/ formula}
Our analysis of eternal solutions is motivated by a classification
theorem of Bertoin  for Smoluchowski's coagulation equation with
additive kernel~\cite{B_eternal}. Here eternal solutions were shown to
be in correspondence with divergent \smeass\/. This theorem was
generalized to other solvable kernels in~\cite{MP3}, based on the
observation that there is a  
natural Laplace  exponent $\lap_t$ associated to every
solution. For the model now under study, it is determined by the \smeas
\be\label{Gdef}
\sm_t(dx) = \frac{x F_t(dx)}{t\kappa^\#(t)},
\ee
and the associated Laplace exponent is 
\be
\label{eternal1}
\lap_t(q) = \frac{ 1 - \Rbar_t(q) }{t\kappa^\#(t)} 
= \int_0^\infty \frac{1 - e^{-qx}}{x} \sm_t (dx). 
\ee
(Recall that $\kappa^\#$ from \qref{kapsharp} 
is the de Bruijn conjugate of $\kappa$ from \qref{kappadef},
and that $\kappa^\#(0+)=\infty$ if $\sum p_kk\log k=\infty$,
$\kappa^\#(0^+)<\infty$ if $\sum p_kk\log k<\infty$.)
\begin{thm} \label{thm:LK}
\begin{enumerate}
\item[(a)]Let $\newrho$ be an eternal solution of
\qref{newmoment}. Then there is a
  divergent \smeas\ $\dsm$ such that 
$\sm_t$ converges to $\dsm$ as $t \dnto 0$. 
\item[(b)] Conversely, for every divergent \smeas\ $\dsm$, 
there is a unique eternal solution $\newrho$ of \qref{newmoment}
  such that $\sm_t$ converges to $\dsm$ as $t \dnto 0$.  
\item[(c)] The Laplace exponent of $\dsm$ is related to the min history
$\len(t)$ by
\be
\label{lk1}
\log \lapdsm(q)= \int_0^{1} \left( 1-e^{-q\len(s)}
\right) \frac{ds}{s} - \int_{1}^\infty e^{-q\len(s)} \frac{ds}{s}.
\ee
\end{enumerate}
\end{thm}

To fix ideas, it may help to note that the self-similar solutions
are generated by the power-law Laplace exponents
\be
\label{eternal4}
\lapdsm(q) = {e^{\theta\gamma}}q^\theta, \quad \theta \in (0,1], 
\ee
where $\gamma$ is the Euler-Mascheroni constant, which satisfies \cite{Olver} 
\be
\label{eternal5}
\gamma = \int_0^1 \frac{1-e^{-s}}{s}\,ds - \int_1^\infty 
\frac{e^{-s}}{s} \,ds.
\ee
This normalization is chosen so the min 
\be
\label{eternal6}
\len_\theta(t) = t^{1/\theta}, \quad \theta \in (0,1]. 
\ee
The corresponding divergent \smeass\/ are given by
\be
\label{eternal7}
\dsm_\theta(x) = \frac{\theta e^{\theta\gamma}}{ \Gamma(2-\theta)}
x^{1-\theta}, \quad \theta \in (0,1].
\ee
\begin{proof}[Proof of Theorem~\ref{thm:LK}.]
In all that follows, $q>0$
is fixed, and we use the equivalence between convergence of \smeass\/
and pointwise convergence of Laplace exponents, as established in
\cite[Sec.\ 3]{MP3}, for example.

{\em 1.\/} First, assume that $\newrho$ is an eternal solution.
We claim that for all $q>0$, $\eta_t(q)\to\lapdsm(q)$ as $t\dnto0$, 
where $\lapdsm$ is given by \qref{lk1} and satisfies
$\lapdsm(\infty)=\infty$.
By \cite[Sec.\ 3]{MP3} it follows that $\lapdsm$ is the Laplace
exponent of a divergent \smeas\ $H$ and $G_t\to H$ as $t\dnto0$.

{\em 2.\/} We first recall from Theorem~\ref{thm.wp} that the min history
satisfies \qref{Lest} whenever $0<t_0<t$. Regarding $t$ as fixed and
$t_0$ variable, we conclude that 
\be\label{Lest2}
0<\len(s)< C s^{(\log 2)/Q_1}, \quad 0<s<1.
\ee
It follows from this, the estimate $1-e^{-q\len(s)}\le q\len(s)$,
and \qref{Lest} that the integrals in \qref{lk1} converge, so that
$\lapdsm(q)$ is finite for $0<q<\infty$.
Moreover, $\lapdsm(0)=0$ and  $\lapdsm(\infty)= \int_0^1 ds/s=\infty$.

{\em 3.\/} We let $w=w_t(q)=1-\Rbar_t(q) = t\kappa^\#(t)\eta_t(q)$ and use
the solution formula \qref{newsol5} together with \qref{kappadef}  
to write
\be\label{weq1}
\log(w\kappa(w)) = -\varphi(\Rbar_t(q)) = -\int_t^\infty
e^{-q\len(s)}\frac{ds}s.
\ee
Adding $-\log t=\int_t^1ds/s$ to both sides we find
\be
\log(w\kappa(w)/t) = 
\int_t^{1} \left( 1-e^{-q\len(s)}
\right) \frac{ds}{s} - \int_{1}^\infty e^{-q\len(s)} \frac{ds}{s}
=\log \lapdsm(q) +o(1)
\ee 
as $t\dnto0$. Hence $w\kappa(w)\sim t \lapdsm(q)$, and asymptotic
inversion yields
\be
w \sim t \lapdsm(q) \kappa^\#(t \lapdsm(q)) \sim 
 t \lapdsm(q) \kappa^\#(t)
\ee
since $\kappa^\#$ is slowly varying. But immediately this yields
$\eta_t(q)\to \lapdsm(q)$ as $t\dnto0$, and this finishes the proof of
(a).

{\em 4.\/} We now establish the converse. Let $\dsm$ be a 
divergent \smeas\ with Laplace exponent $\lapdsm$ (not known at first
to satisfy \qref{lk1}). We  will first establish that the
\LK\/ formula \qref{lk1} defines an appropriate min history $\len$,
then verify that $\len$ defines an eternal solution. 

First, we remark that the definition of the trace admits a 
natural modification for eternal
solutions. Given the min history of an eternal solution, we define the
trace through \qref{defA2}. Conversely, we say $A$ is a {\em maximal
  trace\/} if $A$ is a distribution function on $(0,\infty)$ such that
(i) $\lim_{\tau \to 0} A(\tau)=-\infty$, (ii) $\lim_{\tau \to \infty}
A(\tau)=\infty$, and (iii) $\int_0^{\tau_0} \tau A(d\tau) <\infty$ for
some $\tau_0 >0$. In this case, the min-history is given by
\qref{defL2} with $t_0=0$.

Since $\lapdsm$ is the Laplace exponent of a \smeas, $\lapdsm'$
and $1/\lapdsm$ are completely monotone functions. Thus, there is a
positive measure $A$ such that
\be
\label{eta12}
(\log \lapdsm)' = \frac{\lapdsm'}{\lapdsm} = \int_0^\infty e^{-q\tau}
\tau A(d\tau).
\ee
For every $\tau_0>0$, the measure $A$ satisfies the finiteness conditions
\be
\label{finite1}
(\log \lapdsm) '  \geq  \left\{ \begin{array}{ll} e^{-q\tau_0}
  \int_0^{\tau_0} \tau  A(d\tau), \\
\tau_0 \int_{\tau_0} ^\infty e^{-q\tau} A(d\tau). \end{array}\right.
\ee
Therefore, we may integrate \qref{eta12} between $q$ and $q_1 \in
(0,\infty)$ and rearrange terms to obtain 
\ba
\label{eta13}
\lefteqn{ \log \lapdsm(q) - \int_0^{\tau_0} (1-e^{-q\tau})
    A(d\tau) + \int_{\tau_0}^\infty e^{-q\tau} A(d\tau)} 
\\
\nn
&& = 
\log \lapdsm(q_1) - \int_0^{\tau_0} (1-e^{-q_1\tau})
    A(d\tau) + \int_{\tau_0}^\infty e^{-q_1\tau} A(d\tau):=C(q_1,\tau_0).
\ea
The right hand side is independent of $q$. We let $q \to \infty$ on
the left hand side, and use $\eta(\infty)=\infty$ to see that
$\int_0^{\tau_0} A(d\tau)=\infty$. Similarly, we let $q \to 0$ and use
$\eta(0)=0$ to see that $\int_{\tau_0}^\infty A(d\tau)=\infty$. Thus,
$A$ defines a maximal trace. Let $\len$ denote the associated min
history given by $\len = \exp (A^{\dag})$. 

{\em 5.\/} Since $\tau_0>0$ is arbitrary, we may suppose $\tau_0$ is
in the range of $\len$ and $\len(t_0)=\tau_0$. We change variables in 
\qref{eta13} to obtain 
\be
\label{eta14}
\log \lapdsm(q) - \int_0^1 \left(1-e^{-q\len(s)}\right) \frac{ds}{s} +
  \int_1^\infty e^{-q\len(s)} \frac{ds}{s} = C+\log t_0.
\ee
In order to obtain the \LK\/ formula in the form \qref{lk1}, we
just replace $\len$ by the rescaling $\hat\len(s)=\len(as)$ where $\log a=
C+\log t_0$.

{\em 6.\/} It remains to check that the solution formula
\qref{newsol5} defines a solution for every $t>0$. This is proven by
an approximation argument. We consider a sequence of finite \smeass\/
$\sm_n$ that converge to the divergent \smeas\/ $\dsm$.  
We may suppose that $c_n:=\int_0^\infty x^{-1}\sm_n(dx) \geq n$. 
Let $t_n=1/c_n$ so that $0 < t_n\leq n^{-1}$. 
We will show that the sequence of solutions $\newrho\nsup$ 
defined for $t\ge t_n$, with
initial data given by the probability measures 
\[ F_n(dx):= t_n\kappash(t_n)x^{-1}\sm_n(dx), \]
converges to a solution $\newrho$ satisfying \qref{newsol5}. 

{\em 7.\/}
Let $\eta_n$ be the Laplace exponent of $\sm_n$; 
then $\eta_n(q)\to\lapdsm(q)$ for $q>0$. Define
\be\label{wndef}
w_n(q) := 1-\Rbar_n(q) = t_n\kappash(t_n)\eta_n(q).
\ee
This is related to the min history $\len_n$ determined from $F_n$ 
by the solution formula as in \qref{weq1}, namely,
\be\label{wneq}
\log( w_n\kappa(w_n)/t_n) = 
\int_{t_n}^{1} \left( 1-e^{-q\len_n(s)}
\right) \frac{ds}{s} - \int_{1}^\infty e^{-q\len_n(s)} \frac{ds}{s}. 
\ee
Since $w_n(q)\sim t_n\eta_n(q) \kappash(t_n\eta_n(q))$ as
$n\to\infty$, asymptotic inversion yields $w_n\kappa(w_n)\sim
t_n\eta_n(q)$, and then  \qref{wneq} yields
\be
\log\lapdsm(q) 
 = \lim_{n \to \infty}
\int_{t_n}^1 \left(1 -e^{-q\len\nsup(s)} \right) \frac{ds}{s} -
\int_1^\infty e^{-q\len\nsup(s)}\frac{ds}{s}.
\ee
Convergence of completely monotone functions also implies convergence of all
derivatives. Thus,
\[ \frac{\lapdsm'}{\lapdsm} = \lim_{n \to \infty}
\int_{t_n}^\infty e^{-q\len\nsup(s)} \frac{\len\nsup(s)}{s} \,ds. \]
It follows that the inverse functions $\len\nsup \to \len$ at all points of
continuity. Therefore, we may let $n \to \infty$  in the solution formula 
\[ \Rbar_t\nsup(q) = \forma^{-1} \left(\int_t^\infty
e^{-q\len\nsup(s)}\frac{ds}{s}\right) \]
to see that $\newrho$ defines an eternal solution.
\end{proof}

\section*{Acknowledgements}
This material is based upon work supported by the National Science
Foundation under grant nos.\ 
DMS 06-04420, DMS 06-05006, DMS 07-48482,
and by the Center for Nonlinear Analysis under NSF grants 
DMS 04-05343 and 06-35983.
RLP and BN thank the DFG for partial support through a Mercator
professorship for RLP at Humboldt University and through the Research Group
{\it Analysis and Stochastics in Complex Physical Systems}.

\bibliography{mnp1}

\begin{thebibliography}{10}

\bibitem{B_eternal}
{\sc J.~Bertoin}, {\em Eternal solutions to {S}moluchowski's coagulation
  equation with additive kernel and their probabilistic interpretations}, Ann.
  Appl. Probab., 12 (2002), pp.~547--564.

\bibitem{BlG}
{\sc J.~Bertoin and J.-F. Le~Gall}, {\em Stochastic flows associated to
  coalescent processes. {III}. {L}imit theorems}, Illinois J. Math., 50 (2006),
  pp.~147--181 (electronic).

\bibitem{Bingham}
{\sc N.~H. Bingham, C.~M. Goldie, and J.~L. Teugels}, {\em Regular variation},
  vol.~27 of Encyclopedia of Mathematics and its Applications, Cambridge
  University Press, Cambridge, 1987.

\bibitem{CP}
{\sc J.~Carr and R.~Pego}, {\em Self-similarity in a coarsening model in one
  dimension}, Proc. Roy. Soc. London Ser. A, 436 (1992), pp.~569--583.

\bibitem{deHaan}
{\sc L.~de~Haan}, {\em An {A}bel-{T}auber theorem for {L}aplace transforms}, J.
  London Math. Soc. (2), 13 (1976), pp.~537--542.

\bibitem{Derrida}
{\sc B.~Derrida, C.~Godr\'{e}che, and I.~Yekuitieli}, {\em Scale-invariant
  regimes in one-dimensional models of growing and coalescing droplets}, Phys.
  Rev. A, 44 (1991), pp.~6241--6251.

\bibitem{Feller}
{\sc W.~Feller}, {\em An introduction to probability theory and its
  applications. {V}ol. {II}.}, Second edition, John Wiley \& Sons Inc., New
  York, 1971.

\bibitem{GM}
{\sc T.~Gallay and A.~Mielke}, {\em Convergence results for a coarsening model
  using global linearization}, J. Nonlinear Sci., 13 (2003), pp.~311--346.

\bibitem{IKR}
{\sc I.~Ispolatov, P.~L. Krapivsky, and S.~Redner}, {\em War: {T}he dynamics of
  viscious civilizations}, Phys. Rev. E, 54 (1996), pp.~1274--1289.

\bibitem{MP1}
{\sc G.~Menon and R.~Pego}, {\em Approach to self-similarity in
  {S}moluchowski's coagulation equations}, Comm. Pure Appl. Math., 57 (2004),
  pp.~1197--1232.

\bibitem{MP3}
{\sc G.~Menon and R.~L. Pego}, {\em The scaling attractor and ultimate dynamics
  for {S}moluchowski's coagulation equations}, J. Nonlinear Sci., 18 (2008),
  pp.~143--190.

\bibitem{NK}
{\sc T.~Nagai and K.~Kawasaki}, {\em Statistical dynamics of interacting kinks
  {II}}, Physica A, 134 (1986), pp.~482--521.

\bibitem{NP1}
{\sc B.~Niethammer and R.~L. Pego}, {\em Non-self-similar behavior in the {LSW}
  theory of {O}stwald ripening}, J. Statist. Phys., 95 (1999), pp.~867--902.

\bibitem{Olver}
{\sc F.~W.~J. Olver}, {\em Introduction to Asymptotics and Special Functions},
  Academic Press, New York, 1974.

\end{thebibliography}

\end{document}